\documentclass[11pt,letterpaper]{article}
\usepackage[T1]{fontenc}
\usepackage{graphicx}

\usepackage{geometry}
\usepackage{amsfonts}
\usepackage{microtype}
\usepackage{xcolor}
\usepackage{amsmath}

\usepackage{tikz}
\usepackage{enumitem}
\usepackage{bm}

\usepackage[ruled,vlined,linesnumbered]{algorithm2e}
\providecommand{\DontPrintSemicolon}{\dontprintsemicolon}

\SetCommentSty{mycommfont}

\usepackage{thmtools}
\usepackage{thm-restate}

\usepackage{url}
\usepackage[hidelinks]{hyperref}
\usepackage[utf8]{inputenc}
\usepackage[small]{caption}
\usepackage{booktabs}
\usepackage{algorithmic}
\usepackage[switch]{lineno}

\usepackage{amsthm}
\newtheorem{theorem}{Theorem}[section]
\newtheorem{lemma}[theorem]{Lemma}
\newtheorem{definition}[theorem]{Definition}
\newtheorem{claim}[theorem]{Claim}
\newtheorem{corollary}[theorem]{Corollary}
\newtheorem{proposition}[theorem]{Proposition}
\newtheorem{observation}[theorem]{Observation}
\theoremstyle{definition}
\newtheorem{remark}[theorem]{Remark}
\theoremstyle{plain}

\usepackage{titling}
\usepackage[noblocks]{authblk}

\colorlet{revcolor}{red}

\colorlet{revcolorx}{black}
\newcommand{\newx}[1]{{\color{revcolorx}#1}}
\newenvironment{revisedx}{\color{revcolorx}}{}

\newcommand{\hv}{\hat{v}}

\begin{document}

\title{Envy Cycle Elimination with Strategic Agents:\\ Best Responses and Fairness Guarantees}

\author[1,2]{Georgios Amanatidis}
\affil[1]{Department of Informatics, Athens University of Economics and Business, Greece.}
\affil[2]{Archimedes / Athena RC, Greece.}
\author[3]{Georgios Birmpas}
\affil[3]{Department of Computer Science, University of Liverpool, UK.}
\author[4]{Rebecca Reiffenh{\"a}user}
\affil[4]{Institute for Logic, Language and Computation, University of Amsterdam, The Netherlands.  }

\predate{}
\postdate{}
\date{}

\maketitle

\begin{abstract}
With strong evidence in the literature showing that fairness and truthfulness are incompatible, there is a recent line of work focusing on the fairness properties of equilibria of simple fair division mechanisms, especially Round-Robin. We consider the Envy Cycle Elimination (E-C-E) procedure of Lipton et al.~\cite{LiptonMMS04}, one of the most versatile tools in fair division. While this simple and intuitive algorithm achieves allocations that are envy-free up to one item (EF1) for any number of agents and general monotone valuation functions, surprisingly little is known about its behavior when agents act strategically. We demonstrate how the presence of incentives, although highly natural and relevant for the majority of applications, completely removes the intuitive clarity in the algorithm's execution, even for a few agents and very simple valuation functions. Additionally, while in the standard algorithmic setting there is great flexibility in how the details of E-C-E are implemented, here additional specifications are needed before the procedure is clearly defined, the choice of which has a potentially huge impact on the agents' behavior. Despite these obstacles, for various natural versions of E-C-E, we give the first results on the existence of Pure Nash Equilibria of the resulting mechanisms, and show there exist versions where fairness guarantees are approximately preserved for agents who play best responses\newx{, as long as there are at most three agents; for four or more agents, no such guarantee can hold for every best response, even at a pure Nash equilibrium. Our guarantees extend to monotone subadditive valuation functions. Moreover, best-responding additive agents always secure a constant fraction of their maximin share in the variants where the selected agent chooses her item}.
\end{abstract}

\section{Introduction}\label{sec:intro}

The fair allocation of resources is a central problem in Economics that has also attracted significant attention in Mathematics, Political Science, and Computer Science. Over, roughly, the last two decades, there is a flourishing literature defining and exploring a rich landscape of different settings. One of the most studied variants of the problem considers a set of indivisible, desirable resources (\emph{goods}) that must be allocated to a set of agents, each of whom has an additive (or more general) valuation function over the resources.  We study this setting under the assumption that agents are \emph{strategic}.

Although great progress has been made on its non-strategic counterparts, the strategic setting is much less explored. The main reason is that truthfulness and fairness are proven to be incompatible, even for restricted settings and across multiple fairness criteria \cite{Papai00,Klaus02,Klaus03,AmanatidisBCM17}. More closely investigating the connection between fair division and incentives, some recent works explore the fairness guarantees of non-truthful mechanisms in their equilibria \cite{ABFLLR24,AmanatidisBLLR25,BirmpasELR24} or try to bound the gain an agent may obtain from non-truthful deviations \cite{TY24}.

Unsurprisingly, these attempts focus on very structured mechanisms, where agents have a somewhat restricted strategy space. 
The only relatively well-understood case is the renowned Round-Robin algorithm viewed as a one-shot mechanism (i.e., agents provide their whole input right from the beginning). 
Here, pure Nash equilibria always exist \cite{GW17} and best responses induce allocations that have fairness guarantees with respect to the true valuation functions \cite{ABFLLR24}, with results extending to submodular valuation functions, as well as to subadditive valuation functions that are also cancelable for approximate equilibria \cite{AmanatidisBLLR25}.
However, focusing on Round-Robin severely limits the reach of these results in two ways. First, its proven fairness guarantees fail to extend beyond submodular valuation functions even in the non-strategic setting. Second, one isolated algorithm preserving fairness when strategic agents choose \emph{good} reports does not suffice to reach a deeper understanding of the interplay between EF1 fairness and agents' desire to maximize their utility. Due to these limitations, one of the most natural questions to ask next is whether (some version of) the celebrated Envy-Cycle-Elimination algorithm (E-C-E) of Lipton et al. \cite{LiptonMMS04} exhibits equally nice properties.

Technically, E-C-E is not a single algorithm, but rather a whole family of them. Depending on how the next agent  and the next item is chosen, or when envy cycles are resolved and in which order, one gets different variants that all output EF1 allocations for monotone valuation functions  \cite{LiptonMMS04}, although some might have additional properties \cite{BarmanK20,CaragiannisGRSV23}.
Unfortunately, even with additive valuations, given as input before the mechanism starts, the games induced by E-C-E variants are vastly more complicated than their counterparts for Round-Robin. The main reason is that an agent can affect the order in which the agents may receive goods by (un)blocking some of them for a number of rounds in the underlying envy graph. The effects of this complex behavior are also reflected in the infinite \emph{incentive ratio} of such mechanisms. As recently shown by Tao and Yang \cite{TY24} for two very natural E-C-E variants \newx{(Standard-E-C-E herein, and a variant where, as in Best-Item-E-C-E herein, the selected agent receives her favorite remaining item, but envious agents are not given priority)}, there are instances for which an agent may arbitrarily improve her utility by deviating from her true bid.

Despite this intricate behavior, we initiate the study of E-C-E and its fairness properties when an agent submits
a valuation function that is a best response to the other agents' bids.
Of course, the more structured version of this question is to assume that all agents
play a best response simultaneously, i.e., consider pure Nash equilibria. However,
the existence of such equilibria is far from trivial and, as we show, not always guaranteed.
Nevertheless, we prove that, for a small number of agents,
even isolated best responses induce fair allocations from the (true) perspective of the corresponding agents\newx{, and that the restriction to a small number of agents is unavoidable}.

\medskip
\noindent\textbf{Our Contributions.\ } To the best of our knowledge, our work is the first that studies the stable states of the celebrated E-C-E algorithm, under the setting of strategic agents. In particular, our paper focuses on the questions of whether pure Nash equilibria always exist, and whether the fairness properties of the algorithm are preserved for agents who play a best response to the other agents' strategies. We consider several natural variants of E-C-E, and  provide both positive and negative results\newx{, which demonstrate how strongly the choice of the variant affects what can be guaranteed} by the algorithm---viewed as a mechanism---in this setting. More specifically:
\begin{itemize}[itemsep=2pt, topsep=4pt, leftmargin=9pt]
\item In Section \ref{sec:PNE}, we explore the frontier of (non-)existence of pure Nash equilibria. We show that for several versions of E-C-E, pure Nash equilibria always exist for very simple scenarios with two agents who have unary additive valuation functions, and that, for every version, they exist when all agents have identical binary valuation functions and the number of valuable items is a multiple of $n$. At the same time, the existence no longer holds, even for extremely simple cases with two agents and at most four items, namely for unary valuation functions under Best-Item-E-C-E and for identical binary valuation functions under Standard-E-C-E and Priority-E-C-E. Interestingly, the instances implying this non-existence do have pure Nash equilibria under Preferred-Item-E-C-E, the variant where agents also report the order in which they would like to receive the items.
\item In light of this impossibility, we study whether the fairness properties of the algorithm are preserved for agents who play (pure) best response strategies. In Section \ref{sec:BR2}, we focus on two agents\newx{, and show that, for most versions of E-C-E considered, whoever plays a best response is guaranteed to see the resulting allocation as EF1 (according to her true valuation function!) when her valuation function is additive. For every version, she sees it as $\frac12$-EF1 even when her valuation function is monotone subadditive.}
\item \newx{Going beyond the case of two agents, in Sections \ref{sec:BRn} and \ref{sec:BR3}, we first observe that, for $n\ge3$ agents, a single deviation, namely truth-telling, guarantees that any best response is $\frac{1}{2n-3}$-EF1 from the agent's perspective in every variant ($\frac{1}{2n-2}$-EF1 for monotone subadditive valuation functions), unless two or more single items are worth more to the agent than her whole bundle. We then show that for three agents, best responses might fail to induce approximately EF1 allocations even in the simplest E-C-E variant we consider. Nevertheless, we prove that in the two versions of E-C-E where the selected agent chooses her item, anyone playing a best response strategy is guaranteed to see the resulting allocation as being $\frac{1}{3}$-EF1---or $\frac14$-EF1 when her valuation function is monotone subadditive.}
\item \newx{In Section \ref{sec:BR4}, we show that these positive results cannot be extended to four or more agents: for every $n\ge 4$ and every variant we consider, there is a pure Nash equilibrium whose allocation is not $\alpha$-EF1, for any $\alpha>0$, from the perspective of one of the agents.}
\item \newx{Our results for monotone subadditive valuation functions go qualitatively beyond previous work: to the best of our knowledge, they are the first fairness guarantees for best-responding agents with general monotone subadditive valuation functions in this line of work, as the guarantees of \cite{AmanatidisBLLR25} beyond submodular functions require the valuation functions to be cancelable as well.} 
\item \newx{Finally, in Section \ref{sec:BR-MMS}, we turn to maximin share fairness. We show that best-responding additive agents are guaranteed a $\frac{n}{2n-1}$ fraction of their maximin share under Best-Item-E-C-E and Preferred-Item-E-C-E, and a $\frac1n$ fraction under every variant, irrespectively of what the other agents report.}
\end{itemize}
\medskip

\newx{\noindent Table~\ref{tab:results} summarizes our results on best responses.}

\begin{table}[h]
\centering\small
\renewcommand{\arraystretch}{1.3}
\begin{tabular}{lcccc}
\toprule
 & Standard & Priority & Best-Item & Preferred-Item\\
\midrule
$n=2$, additive & EF1 & EF1 & $\frac12$-EF1 & EF1 \\
$n=2$, subadditive & $\frac12$-EF1 & $\frac12$-EF1 & $\frac12$-EF1 & $\frac12$-EF1\\
$n=3$, additive & no approx. & ? & $\frac13$-EF1 & $\frac13$-EF1 \\
$n=3$, subadditive & no approx. & ? & $\frac14$-EF1 & $\frac14$-EF1\\
$n\ge4$ & no approx. & no approx. & no approx. & no approx.\\
$n\ge3$, at most one large item & $\frac{1}{2n-3}$-EF1 & $\frac{1}{2n-3}$-EF1 & $\frac{1}{2n-3}$-EF1 & $\frac{1}{2n-3}$-EF1\\
$n\ge 2$, additive & $\frac1n$-MMS & $\frac1n$-MMS & $\frac{n}{2n-1}$-MMS & $\frac{n}{2n-1}$-MMS\\
\bottomrule
\end{tabular}
\caption{\newx{Guarantees from the perspective of a best-responding agent $i$. ``No approx.'' means that some best response, even at a pure Nash equilibrium for $n\ge 4$, is not $\alpha$-EF1 from $i$'s perspective for any $\alpha>0$. An item is large if $i$ values it more than her final bundle; the corresponding guarantee is $\frac{1}{2n-2}$-EF1 for subadditive valuation functions. The $\frac1n$-MMS guarantee also holds for monotone subadditive valuation functions.}}\label{tab:results}
\end{table}

\smallskip
\noindent\textbf{Related Work.\ }The literature on fair division, even restricted to indivisible item settings, has recently exploded. We restrict ourselves to only mention works that involve incentives; the interested reader is referred to a number of excellent surveys \cite{Procaccia_cake_16,LindnerR16,moulin2019fair,LangR16,BouveretCM16,AmanatidisABFLMVW23} for a comprehensive overview of the area.

Existing contributions investigating the design of truthful mechanisms to allocate a set of indivisible items to a set of strategic agents \cite{Papai00,Klaus02,Klaus03}, suggest that such mechanisms are very restricted---often just dictatorships---and cannot have any non-trivial fairness guarantees. Even in the additive case, and for just two agents, no truthful mechanism can have a constant-factor approximation for any commonly used fairness notion \cite{MarkakisP11,AmanatidisBM16,AmanatidisBCM17}. The only known positive results on truthful mechanisms with meaningful fairness guarantees impose massive restrictions on the number of agents and items \cite{xCKKK09}, or the structure of valuation functions \cite{HPPS,BabaioffEF21}.

Due to these impossibilities, other approaches have been considered. For instance, there are works that study Bayesian incentive compatible design \cite{Gkatzelis0TV24} or fairness in conjunction with relaxed versions of truthfulness \cite{PV22}. The line of work closest to ours studies pure Nash equilibria of one of the most fundamental and simple algorithms used in fair division of indivisible items, namely Round-Robin \cite{GW17,AmanatidisBLLR25,ABFLLR24}. In particular, Aziz et al. \cite{GW17} show the existence of pure Nash equilibria, and Amanatidis et al. \cite{AmanatidisBLLR25,ABFLLR24} prove that every pure Nash equilibrium of Round-Robin (approximately) preserves the fairness properties of the mechanism, not just for additive valuation functions, but also for superclasses of them. A similar approach was then followed by Birmpas et al. \cite{BirmpasELR24}, in a generalized setting where the agents' valuation functions can be interdependent. To the best of our knowledge, the only other work that considers the E-C-E under the strategic setting is that of Tao and Yang \cite{TY24}, where the performance of E-C-E is explored with respect to the (worst-case) ratio between the value an agent may get via manipulation and her value under truth-telling.
\newx{Regarding maximin share fairness, variants of E-C-E are known to achieve constant-factor approximations of the maximin share in the non-strategic setting, besides EF1 or approximate EFX guarantees (see, e.g., \cite{AmanatidisABFLMVW23}).}

\section{Preliminaries}\label{sec:prelim}
Let $N = [n] := \{1, 2, \ldots, n\}$ be a set of $n$ agents and $M = \{g_1,\ldots, g_m\}$ be a set of $m$ items.
We consider the problem of fully and \textit{fairly} allocating the set $M$ to the agents of $N$ under the assumption that the agents are strategic.
Each agent $i\in N$ has a valuation function $v_i:2^M\to \mathbb{R}_{\ge0}$ over the subsets of items\newx{, with $v_i(\emptyset)=0$}. That is, the items are \textit{goods} as they always have nonnegative (marginal) value.
The valuation functions in this work are all \emph{monotone subadditive}, i.e.,  $v_i(S) \le v_i(S\cup T) \le v_i(S)+ v_i(T)$ for every $S, T \subseteq M$, but
for most of our results they are assumed to be (possibly a special case of) \textit{additive}, i.e., $v_i(S\cup T) = v_i(S) + v_i(T)$ for every $S, T \subseteq M$ with $S\cap T= \emptyset$. For the sake of readability, we write $v_i(S-g)$ instead of $v_i(S\setminus\{g\})$ and $v_i(g)$ instead of $v_i(\{g\})$. When $v_i$ is additive and $v_i(g)\in \{0,1\}$ (resp.~$v_i(g) = 1$) for all $g\in M$, we say that $v_i$ is \textit{binary} (resp.~\textit{unary}).
Note that an additive function $v_i$ can be fully specified by the vector $(v_i(g_1), v_i(g_2), \ldots, v_i(g_m))$.

An \textit{allocation} $A =  (A_1,\ldots,A_n)$ of the items is a partition of $M$ into $n$ sets;  $A_i$ is called the \emph{bundle} of agent $i$. We next define
\emph{envy-freeness up to one item} (EF1) \cite{LiptonMMS04,Budish11}, a relaxation of \emph{envy-freeness} \cite{GS58,Foley67,Varian74}, which is the main fairness notion in this work.

\begin{definition}\label{def:EF-EFX}
	Let $\alpha \in [0,1]$. An allocation $(A_1,\ldots,A_n)$ is $\alpha$-\textit{envy-free up to one item} ($\alpha$-EF1), if for any pair of agents $i, j\in N$, with $A_j\neq\emptyset$, there exists an item $g\in A_j$, such that $v_i(A_i) \geq \alpha \cdot v_i(A_j - g)$.
\end{definition}

Occasionally, we deal with allocations that are not $\alpha$-EF1 but are such that for every $j\in N$ with $A_j\neq\emptyset$, there is some good $g_{(j)}\in A_j$, such that $v_i(A_i) \geq \alpha \cdot v_i(A_j - g_{(j)})$. We then say that $(A_1,\ldots,A_n)$ is $\alpha$-EF1 \emph{from agent $i$'s perspective}.

\paragraph{Mechanisms}\label{subsec:mechanisms}
In our setting, where there are no payments, \emph{mechanisms} essentially refer to algorithms that receive their input from the agents.
So, the main distinction from algorithms is  that the reported valuations may differ from the true ones.
Throughout this work we assume that each agent $i$ reports a \emph{bid vector} $\bm{b}_i$ specifying her valuation function as well as a \textit{preference list} dictating the order in which she would prefer to get the items. The latter is relevant \newx{only for Preferred-Item-E-C-E (see Section~\ref{subsec:e-c-e})}, and is not required to be consistent with the submitted values. For the additive setting, this means that  $\bm{b}_i = (b_{i1}, b_{i2}, \ldots, b_{im}; \textrm{pref}_i)$, where $b_{ij}\ge 0$ is the value agent $i$ claims to have for item $g_j$ and $\textrm{pref}_i$ is a permutation of $M$. For the subadditive case, the first part of the bid vector instead consists of a value oracle \newx{ for a monotone subadditive function $f_i$ with $f_i(\emptyset)=0$, i.e., $\bm{b}_i = (f_i; \textrm{pref}_i)$. We write $\hv_i$ for the valuation function \emph{reported} by agent $i$ as part of $\bm{b}_i$; that is, $\hv_i(S)=\sum_{g_j\in S}b_{ij}$ in the additive setting and $\hv_i=f_i$ in the subadditive setting. We say that agent $i$ reports \emph{truthfully} if $\hv_i=v_i$, whatever her preference list is.}

A mechanism $\mathcal{M}$ takes as input a \emph{bid profile} $\mathbf{b} = (\bm{b}_1, \bm{b}_2, \ldots, \bm{b}_n)$ of bid vectors and outputs an allocation $\mathcal{M}(\mathbf{b}) = (\mathcal{M}^{(1)}(\mathbf{b}),  \ldots, \allowbreak \mathcal{M}^{(n)}(\mathbf{b}))$. Given a profile $\mathbf{b} = (\bm{b}_1, \allowbreak \ldots, \bm{b}_n)$, we use the shortcut
$\mathbf{b}_{-i}$ for $(\bm{b}_1, \ldots, \bm{b}_{i-1},  \allowbreak \bm{b}_{i+1}, \allowbreak \ldots, \bm{b}_n)$ and
the shortcut $(\bm{b}'_i, \mathbf{b}_{-i})$ for the profile $(\bm{b}_1,  \allowbreak \ldots, \bm{b}_{i-1}, \allowbreak \bm{b}'_{i}, \allowbreak \bm{b}_{i+1}, \ldots, \bm{b}_n)$, where $\bm{b}'_{i}$ is a bid vector.

\begin{definition}\label{def:PNE}
	Let $\mathcal{M}$ be an allocation mechanism and consider a profile
	$\mathbf{b} = (\bm{b}_1, \ldots, \bm{b}_n)$. We say that $\bm{b}_{i}$ is a \emph{best response} to $\mathbf{b}_{-i}$ if for every bid vector $\bm{b}'_i$, we have
	$v_i(\mathcal{M}^{(i)}(\bm{b}'_i, \mathbf{b}_{-i}))\le v_i(\mathcal{M}^{(i)}(\mathbf{b})) $.
	The profile $\mathbf{b}$ is a \emph{pure Nash equilibrium} (PNE) if, for each $i\in N$, $\bm{b}_{i}$ is a best response to $\mathbf{b}_{-i}$.
\end{definition}

\subsection{Envy Cycle Elimination}\label{subsec:e-c-e}
As we mentioned in the Introduction, the \emph{envy-cycle-elimination algorithm} (E-C-E) of Lipton et al. \cite{LiptonMMS04} is fairly loosely defined and captures a whole family of algorithms. Before formally describing it, we introduce the notion of the \emph{envy graph}.
Suppose we have a partial allocation $P = (P_1,\ldots,P_n)$, i.e., an allocation of a strict subset of $M$. We define the directed envy graph $G_{P} = (N, E_{P})$, where $(i,j) \in E_{P}$ if and only if \newx{$\hv_i(P_i)<\hv_i(P_j)$. Note that the envy graph is defined with respect to the \emph{reported} valuation functions, as these are all the mechanism has access to. To avoid confusion with the true preferences of the agents, we say that agent $i$ \emph{has an edge towards} agent $j$ when $(i,j)\in E_P$, and that $i$ \emph{adds} or \emph{removes} an edge when this changes.} E-C-E iteratively builds the allocation by choosing in each step an agent \newx{who has no incoming edges} (i.e., a \textit{source} in $G_{P}$), and adding to her bundle one of the available items. To ensure such an agent always exists, E-C-E eliminates any emerging cycles, by reallocating the current bundles along the cycle edges.
In particular, if $j_1 \to j_2 \to \ldots \to j_r \to j_1$ is a cycle in $G_{{P}}$, agent $j_k$ receives the bundle previously owned by $j_{k+1}$ (where $j_{r+1}$ is, by convention, agent $j_1$).

\begin{algorithm}[ht]
	\DontPrintSemicolon
	{\small
        $(A_1, \ldots, A_n) \leftarrow (\emptyset, \ldots, \emptyset)$\;
		Construct the envy graph $G_{A}$\;
		\While{$M\neq \emptyset$}{
            Let $i\in N$ be a node of in-degree $0$ \label{line:source} \tcc*{{\scriptsize next source}}
            Let $g\in M$ be an available item \label{line:good} \tcc*{{\scriptsize next item}}
            $A_{i} \leftarrow A_{i}\cup \{g\}$\;
            $M \leftarrow M\setminus \{g\}$\;
			Update $G_{A}$ according to $\mathbf{b}$ \;
            Ensure $G_{A}$ has a source \label{line:resolve} \tcc*{{\scriptsize possible elimination}}
        }

		\Return $A$ \;
	}
	\caption{Envy-Cycle-Elimination$(N, M, \mathbf{b})$
        \label{alg:ece}}
\end{algorithm}

Lines \ref{line:source}, \ref{line:good} and \ref{line:resolve} give a lot of freedom on how the algorithm is actually executed. For instance, the source in line \ref{line:source} might be chosen lexicographically or according to how envious sources are, the item in line \ref{line:good} might be chosen according to a fixed order or according to the preference of the current source, and envy cycles in line \ref{line:resolve} might be resolved immediately or only when no source is left. All these (and several other) possibilities are valid ways to implement E-C-E and guarantee that it produces EF1 allocations in the algorithmic setting. Yet, they greatly affect strategic behaviors in our game theoretic setting. Therefore, we need to fix these choices before we are able to say anything meaningful here.

Starting with the envy cycles, we assume that as long as there are cycles in $G_A$, they are all resolved in line \ref{line:resolve} \newx{one at a time, each time rotating the first cycle with respect to some fixed order over all possible cycles. All our results hold for any such fixed order.}
Then, we consider the following variants, depending on how the next source and good are determined:
\begin{itemize}[itemsep=1pt, topsep=2pt, leftmargin=9pt]
    \item \emph{Standard-E-C-E:}\, Among multiple sources, the next agent is chosen according to a predefined ordering of the agents. The next item is also determined  by a fixed predefined ordering.
    \item \emph{Priority-E-C-E:}\, Among multiple sources, \newx{agents with outgoing edges} have priority; if there are multiple such agents, the next one is chosen according to a predefined ordering of the agents. Like in Standard-E-C-E, the next item is determined by a fixed predefined ordering.
    \item \emph{Best-Item-E-C-E:}\, The next agent is chosen exactly like in Priority-E-C-E. The next item is \newx{the selected} agent's best available item according to the declared values in her bid vector\newx{, i.e., an available item of maximum reported marginal value with respect to her current bundle (where ties are resolved according to a predefined ordering)}.\newx{\footnote{\newx{For subadditive bids, the ``best available item'' is one of maximum reported marginal value $f_i(P\cup\{g\})-f_i(P)$, where $P$ is the current bundle of the selected agent; for additive bids this is just the reported value of $g$. All our results for Best-Item-E-C-E also hold if, instead, the reported value $f_i(\{g\})$ of each single item is used.}}}
    \item \emph{Preferred-Item-E-C-E:}\, Similar to the Best-Item-E-C-E
    above, but now  the next item is \newx{the selected} agent's best available item according to her preference list. Recall that this might be distinct from the available good with the highest (declared) value for her.
\end{itemize}
\newx{Throughout, the predefined ordering of the agents is $1,2,\ldots,n$, and the predefined ordering of the items is $g_1,g_2,\ldots,g_m$.}
At this point it is worth noting that the introduction of $\textrm{pref}_i$ results in a richer strategy space. This additional flexibility the agents have in determining how the items are allocated is critical for our positive results for three agents.

\begin{revisedx}
\smallskip
\noindent\textbf{Runs of E-C-E.\ }
We think of a run of E-C-E as a sequence of \emph{steps} $1,2,\ldots,m$, where in step $s$ a source is selected, an item is selected, the item is added to the bundle of the source, and finally, as long as there is a cycle in the envy graph, the cycles are resolved as described above. In particular, the envy graph has no cycles whenever a source is selected. A \emph{moment} of the run is any point in time between two of these operations. We refer to the bundles of the agents at a given moment as the \emph{bundles present} at that moment. An addition enlarges one bundle, and a cycle resolution only permutes the bundles present among the agents. We call Standard-, Priority- and Preferred-Item-E-C-E \emph{value-independent}: in these variants the selected item depends only on the set of available items, the fixed ordering of the items, and the preference list of the selected agent, but never on reported values.

We will repeatedly use the following basic facts.

\begin{observation}\label{obs:basic}
In every run of every variant of E-C-E we consider:
\begin{enumerate}[label=(\alph*),itemsep=0pt,topsep=2pt]
\item every nonempty bundle present at some moment is contained in exactly one bundle of the final allocation, and distinct nonempty bundles present at the same moment are contained in distinct final bundles;
\item if, when an agent is selected in some step, some agent has an edge towards the holder of a bundle $Z$, then the item of that step is not added to $Z$;
\item for every agent $a$, the reported value $\hv_a$ of $a$'s current bundle never decreases;
\item no agent has an edge towards an agent whose bundle is empty; hence an agent with an empty bundle is a source and does not belong to any cycle;
\item the final allocation is EF1 with respect to the reported valuation functions; in particular, an agent who reports truthfully sees it as EF1 from her perspective.
\end{enumerate}
\end{observation}

\begin{proof}
For part (a), note that at every moment there are exactly $n$ bundles present, an addition enlarges one of them, and a cycle resolution permutes them among the agents. So, following a bundle present at some moment through the subsequent operations, we reach a final bundle that contains it, and two bundles present at the same moment are followed to two distinct final bundles. Since the final bundles are pairwise disjoint, a nonempty bundle is contained in at most one of them. For part (b), note that the holder of $Z$ has an incoming edge, so she is not a source. For part (c), in a cycle resolution every agent on the cycle receives a bundle towards whose holder she had an edge, and an addition does not decrease $\hv_a$, as $\hv_a$ is monotone. For part (d), $\hv_b(\emptyset)=0\le\hv_b(P_b)$ for every agent $b$ and every bundle $P_b$.
Part (e) directly follows by the standard analysis of E-C-E Lipton et al. \cite{LiptonMMS04}, which covers all of our variants. 
\end{proof}
\end{revisedx}

\section{E-C-E and Pure Nash Equilibria}\label{sec:PNE}

In this section we explore whether the several versions of E-C-E that we consider  have pure Nash equilibria. Although we begin with some positive results, these results regard very simple and restricted scenarios with respect to the true valuation functions of the agents. We then  demonstrate that if one tries to go beyond these scenarios, the complex nature of E-C-E (when it is considered under strategic agents) prevents the existence of stable states. We see the set of results in this section as a justification for studying \newx{best responses rather than equilibria} in Section \ref{sec:BR}.

Before we proceed, we want to point out that although the true value profiles that we consider for the agents are restricted, their declarations \textit{are not bounded by this}, i.e., they still can declare any values for the items.

\subsection{Pure Nash Equilibria in Simple Scenarios}\label{sec:PNEun}

The most simple scenario that one can consider regarding the true valuation profiles of the agents is the one where all agents have additive valuation functions, and they value all the items equally, i.e., the valuation functions are \emph{unary} and agents simply strive to maximize the \emph{number} of items assigned to them.
The following positive results, when coupled with our findings in the next section, demonstrate the limits of what one can expect with respect to the existence of pure Nash equilibria.

\begin{proposition}\label{prop:pneun}
    For any of the considered
     E-C-E variants, there exists a pure Nash equilibrium that induces an EF1 allocation according to the true values of the agents, when they have identical binary
    valuation functions, and the number of positively valued items is a multiple of $n$.
\end{proposition}

\begin{proof}
We are considering instances with $n$ agents, all of which have identical binary valuation functions\newx{, and $k\, n$ items of value $1$, for some integer $k\ge 0$}. Notice that as we already mentioned in the beginning of this section, the agents are free to declare whichever value they want for each of the items (their declarations are not affected by the fact that their true values are binary and identical).

We will show that truth-telling, i.e., having all the agents submitting their true values over the items, is a pure Nash equilibrium for all the considered versions of E-C-E. Initially notice that when the agents declare their true values, E-C-E guarantees exactly $k$ items of value 1 to all of them, as otherwise the EF1 condition would be violated. To see this, notice that if there was at least one agent $i$ with more than $k$ items of value 1, this would imply that there would be at least one other agent $i'$ with less than $k$ items of value 1 (as the agents not only have binary valuation functions, but these functions are also identical). As all the agents declared their true values, and agent $i'$ has at least two less items of value 1 than agent $i$, she would not be EF1 towards agent $i$, and this is not possible \newx{by Observation~\ref{obs:basic}(e)}.

Under the same reasoning, if agent $i$ was able to deviate from the truth-telling profile and gain more than $k$ items of value 1, this would once again imply that there would be at least one other agent $i'\neq i$ that would end up with less than $k$ items of value 1 (recall that the values are identical). Since agent $i'$ declared her true values, and there is some agent $i$ that has at least two more items of value 1 than her, she would not be EF1 towards agent $i$, \newx{contradicting Observation~\ref{obs:basic}(e)}.

The property that the pure Nash equilibrium is EF1-fair according to the true values of the agents is derived \newx{directly} from the fact that the agents declare their true values, and the fairness properties of the algorithm.
\end{proof}

\begin{revisedx}
\begin{remark}\label{rem:prop1}
When the number of positively valued items is not a multiple of $n$, truth-telling is not necessarily a pure Nash equilibrium. Consider Standard-E-C-E (or Priority-E-C-E) with two agents and two items, where both agents have value $0$ for $g_1$ and value $1$ for $g_2$. Under truth-telling, agent $1$ receives $g_1$ in the first step; then nobody has an edge, so agent $1$ also receives $g_2$, and agent $2$ ends up with value $0$. If, instead, agent $2$ reports value $1$ for both items, she adds an edge towards agent $1$ after the first step, becomes the unique source, and receives $g_2$. Afterwards she has no edge towards agent $1$, so no cycle resolution happens. In this instance a pure Nash equilibrium does exist (agent $2$ reports value $1$ for both items and agent $1$ reports truthfully, and agent~$1$ cannot prevent agent $2$ from receiving $g_2$), but we do not know whether this is always the case when the number of positively valued items is not a multiple of $n$.
\end{remark}
\end{revisedx}

We proceed with the next theorem that regards the Standard-E-C-E, and the existence of pure Nash equilibria under unary valuation functions, when the number of items is not necessarily a multiple of $n$. Despite the theorem regarding only the case of two agents, the construction  of the equilibrium illustrates the complex behavior of agents even in this simple case.

\begin{theorem}\label{thm:PNE2agents}
    For two agents with unary valuation functions, Standard-E-C-E always has a pure Nash equilibrium that induces an EF1 allocation according to the true valuation functions.
\end{theorem}

\begin{revisedx}
\begin{proof}
If $m$ is even, the theorem follows from Proposition~\ref{prop:pneun}. If $m=1$, agent $1$ receives the only item in every run, as in the first step no agent has an edge and ties are broken in favor of agent $1$; so every profile is a pure Nash equilibrium, and every allocation of a single item is EF1. So, we may assume that $m=2\ell+1$ for some $\ell\ge1$. Recall that Standard-E-C-E allocates the items in the order $g_1,g_2,\ldots,g_m$, so the most recently allocated item is always the one with the largest index allocated so far.

Let agent $1$ report $b_{1j}=1+2^{j-m-1}$ for every item $g_j$, and let agent $2$ report truthfully, i.e., $b_{2j}=1$ for every $g_j$. Since every $2^{j-m-1}$ is larger than the sum of all $2^{r-m-1}$ with $r<j$, and the sum of all these terms is less than $1$, agent $1$ has an edge towards agent $2$ exactly when either $|A_1|<|A_2|$, or $|A_1|=|A_2|$ and the item of largest index allocated so far belongs to $A_2$.

We claim that under this profile, agent $1$ receives $\ell+1$ items and agent $2$ receives $\ell$ items. Indeed, agent $1$ receives $g_1$ in the first step. Whenever $|A_1|=|A_2|+1$, agent $2$ has an edge towards agent $1$ while agent $1$ has no edge, so agent $2$ receives the next item. Then $|A_1|=|A_2|$ and the most recently allocated item belongs to $A_2$; so agent $1$ has an edge towards agent $2$, agent $2$ has no edge towards agent $1$, and agent $1$ receives the next item. Hence the items alternate, starting with agent $1$. The resulting allocation is EF1 according to the true valuation functions.

\emph{Agent $1$ cannot improve.} Agent $2$ reports truthfully, so by Observation~\ref{obs:basic}(e), whatever agent~$1$ reports, agent $2$ ends up with at least $|A_1|-1$ items. Thus, agent $1$ gets at most $\ell+1$ items.

\emph{Agent $2$ cannot improve.} Fix any report of agent $2$. We show that the following invariant holds after every even number of steps (that is, after the cycle resolutions of that step):
\begin{center}
either $|A_1|\ge|A_2|+2$, or $|A_1|=|A_2|$ and agent $2$ has no edge towards agent $1$.
\end{center}
This trivially holds before the first step. Assume it holds after $2r$ steps. We will use the following two facts: agent~$1$ never has an edge when $|A_1|>|A_2|$, and when $|A_1|=|A_2|$ she has an edge exactly if the most recently allocated item belongs to $A_2$.
\begin{itemize}[itemsep=1pt, topsep=2pt, leftmargin=12pt]
\item \emph{$|A_1|=|A_2|$ and agent $2$ has no edge towards agent $1$.} Agent $1$ is a source, and she is selected: she is the unique source if she has an edge towards agent $2$, and she is preferred by the ordering of the agents otherwise. After she receives the item, $|A_1|=|A_2|+1$, so agent $1$ has no edge and no cycle resolution is possible. If agent $2$ has no edge towards agent~$1$ either, agent $1$ is selected again, and $|A_1|=|A_2|+2$. Otherwise agent $2$ receives the next item, the two bundles have the same size, and agent $1$ has an edge towards agent $2$. If agent $2$ still has an edge towards agent $1$, the two bundles are swapped; afterwards each agent holds the bundle towards which she had an edge, so, by the description of agent $1$'s edges above, neither agent has an edge. In every case the invariant holds after $2r+2$ steps.
\item \emph{$|A_1|\ge|A_2|+2$.} Agent $1$ has no edge in the next two steps, as long as $|A_1|>|A_2|$. After the first of the two steps, $|A_1|\ge|A_2|+1$. After the second, either $|A_1|\ge|A_2|+2$, or $|A_1|=|A_2|$ and agent~$2$ received the most recent item. In the latter case, agent $1$ has an edge towards agent $2$, and we conclude exactly as in the previous case.
\end{itemize}
After $2\ell$ steps, if $|A_1|\ge|A_2|+2$, agent $2$ can end up with at most $\ell$ items. Otherwise $|A_1|=|A_2|=\ell$ and agent~$2$ has no edge towards agent $1$; as above, agent $1$ is selected and receives $g_m$. In both cases agent~$2$ ends up with at most $\ell$ items, which is what she gets by reporting truthfully.
\end{proof}
\end{revisedx}

\begin{corollary}\label{cor:2una}
The analog of Theorem \ref{thm:PNE2agents} holds for Priority-E-C-E as well.
\end{corollary}

\begin{proof}
    This follows from the fact that Standard-E-C-E and Priority-E-C-E run identically for instances with only two agents. To see this, notice that the difference between the two versions of the algorithm  is that in the case of Priority-E-C-E, when there are multiple sources, the algorithm breaks ties in favor of the ones that have outgoing edges, while in the case of Standard-E-C-E the ties are being broken lexicographically. However, when $n=2$, the only way of having two sources is to not have outgoing edges from any of the two agents; therefore this additional rule is never put into action.\end{proof}

We would like to point out  that the profile used in the proof of Theorem \ref{thm:PNE2agents} does not necessarily extend to a profile that induces a pure Nash equilibrium for these two versions of the algorithm and for the case of $n \geq 3$ agents. The presence of a third agent might affect the properties described in \newx{the invariant} of our proof, making such a generalization highly non-trivial. On the other hand, as we will see in the next section, the positive results that we  presented so far  more or less define the limits of what one can expect with respect to the existence of pure Nash equilibria.

\subsection{Impossibility Results}\label{sec:PNEnone}

Here we present results regarding the non-existence of pure Nash equilibria, for most of the considered versions of E-C-E. We demonstrate that the positive results that we have presented in the previous section  do not hold even for scenarios with just two agents and four items, when an instance goes beyond unary with respect to the true valuation functions of the agents.

Before going there, in the next theorem we consider the Best-Item-E-C-E version of the algorithm, and we show that it cannot guarantee stability even for the cases that we explored in the previous section.

\begin{theorem}\label{thm:noPNE-Best}
    There are instances where the Best-Item-E-C-E does not have pure Nash equilibria. These instances can be as simple as having just 2 agents with unary valuation functions and 3 items.
\end{theorem}

\begin{proof}
Consider an instance with two agents and three items. The true valuation functions of the agents are unary, i.e., the value of every item for them is equal to 1. We would like to prove that in this case there is no pure Nash equilibrium, i.e., for every possible allocation of the items, at least one of the agents can make a different report and strictly improve her utility. The structure of the argument is as follows: we show that, against every report of the other agent, an agent can secure at least two items unless she already receives two; since there are only three items, in every profile some agent receives at most one item and can improve. To this end, initially notice that allocations where one of the agents get all three items, while the other gets nothing, cannot be an equilibrium as the agent that got no items can always declare her true values (a value of 1 for each of the items) and get at least one item, by Observation~\ref{obs:basic}(e). Therefore, the only cases that remain to examine are the ones where agent 1 gets 2 items and agent 2 gets 1 item, and vice versa. We will show that in both cases, the agent that receives just one item can always change her strategy and get two items.

  Starting with the case where agent 1 gets 2 items and agent 2 gets 1 item, initially notice that regardless of what agent 1 has declared, when she gets her first item (and she is the one that will get the first item as ties are broken lexicographically when there are no outgoing edges), this item is her most favorable one, according to her declaration. This implies that if agent 2 gets the second item (because she added an edge towards the first agent after the latter got the first item), agent 1 cannot add an edge towards agent 2 at that moment. Therefore, agent 2 can declare that the item that agent 1 got has a value 1 for her, and the remaining items have a value of $0.5$. This guarantees that she will add an edge towards agent 1 after agent 1 gets the first item; so she gets the second item, and since the edge that she added remains towards agent 1 (and agent 1 cannot add an edge towards her) she will get the third item as well. Then her reported value is $1$ for each of the two bundles, so she has no edge towards agent 1 and there is no cycle. This strategy guarantees her a total value of 2, which is strictly better than before.

  Moving to the case where agent 2 gets 2 items and agent 1 gets 1 item, consider any declaration of agent 2 that led to an allocation where she gets 2 items. What agent 1 can now do  is to declare as her best item the second best item of agent 2; report value \(1\) for that item and \(0\) for both others (in case there is a tie then it is broken arbitrarily). In that way, since she is the one that gets the first item, agent 2 might have an edge towards her or not, depending on whether agent 2 declared a positive value or a value of 0 for her second best item. In the first case notice that agent 2 gets the next item, but now this item has more or equal value (according to her declaration) than the one that agent 1 got (as she gets her best one), and therefore, after getting this item, she has to remove the edge towards agent 1. In addition, notice that there is also no edge from agent 1 towards agent 2 at this point (for the same reason). Thus, because Best-Item-E-C-E now breaks ties between the agents lexicographically, agent 1 gets the third and final item, and this makes her total value 2, which is strictly better than before. In the second case, once again, as ties are broken lexicographically, agent 1 gets the second item, and since at this point she has no edge towards agent 2 (and will not add one), she ends up with a utility of 2, which is strictly better than before. This completes our proof.
  \end{proof}

  We proceed with the next theorem (and corollary) that considers Standard-E-C-E and Priority-E-C-E, and we show that in instances with two agents it is enough for the agents to have a value of 0 for just one item (and a value of 1 for the rest) for the existence of pure Nash equilibria to be broken.

\begin{theorem}\label{thm:noPNE-Standard}
    There are instances where the Standard-E-C-E does not have pure Nash equilibria. These can be as simple as having just 2 agents with identical binary valuation functions and 4 items.
\end{theorem}

\begin{proof}
    Consider an instance with two agents and four items, where both have binary and identical values for the items. In particular, let $v_i(1)=0$, and $v_i(j)=1$, for every $i \in\{1,2\}$ and $j \in \{2,3,4\}$ (where we slightly abuse the notation and refer to item $j$ rather than $g_j$, for readability). Notice also that if there was a pure Nash equilibrium for this instance, then each agent would get at least one item of value 1, as otherwise she could improve her utility by just declaring her true values\newx{, by Observation~\ref{obs:basic}(e)}. We will split the proof into cases and we will show that it is always possible for an agent to get at least two items from $\{2,3,4\}$, regardless of what the other agent declares. \newx{Since there are only three such items, in any allocation one of the agents gets at most one of them, and this agent has a profitable deviation. Hence no profile is a pure Nash equilibrium. Table~\ref{tab:noPNE} summarizes the four cases.}

\begin{table}[h]
\centering\small
\begin{tabular}{llll}
\toprule
Report of the other agent & Deviating agent & Deviation & Valuable items secured\\
\midrule
agent 1: $b_{12}\ge b_{13}$ & agent 2 & $(0,1,0.5,0.6)$ & $3,4$\\
agent 1: $b_{12}<b_{13}$ & agent 2 & $(1,1,1,1)$ & $2,4$\\
agent 2: $b_{21}=0$ & agent 1 & $(0,0,1,1)$ & two of $2,3,4$\\
agent 2: $b_{21}>0$ & agent 1 & $(0,1,0,1)$ & two of $2,3,4$\\
\bottomrule
\end{tabular}
\caption{\newx{The four cases in the proof of Theorem~\ref{thm:noPNE-Standard}.}}\label{tab:noPNE}
\end{table}

    \begin{itemize}[itemsep=2pt, topsep=3pt, leftmargin=9pt]
        \item Agent 1 makes a declaration where $\newx{b_{12} \geq b_{13}}$. In that case, agent 2 declares \newx{$b_{21}=0$, $b_{22}=1$, $b_{23}=0.5$, and $b_{24}=0.6$}. Then the first two items go to the first agent\newx{: after the first item no agent has an edge, as agent 2 reports value $0$ for item 1, and agent 1 is preferred by the ordering of the agents. After the second item, agent 2 has an edge towards agent 1, while agent 1 does not have an edge towards agent 2, whose bundle is empty}. So, the third item goes to the agent 2. At this point agent 1 does not have an edge towards agent 2\newx{, as $b_{11}+b_{12}\ge b_{12}\ge b_{13}$,} and agent 2 has an edge towards agent 1\newx{, as $0.5<1$}. Therefore, agent 2 gets the fourth item as well, and since at this point she has no edge towards agent 1, the mechanism stops.
        \item Agent 1 makes a declaration where \newx{$b_{12} < b_{13}$}. In that case, agent 2 declares \newx{value $1$ for every item}. Then the first item goes to the agent 1 and the second item goes to agent 2. As agent 2 does not have an edge towards agent 1, agent 1 gets the third item and now she removes any edges (if any) towards agent 2\newx{, as $b_{11}+b_{13}> b_{12}$}. At this point agent 2 has an edge towards agent 1, and thus she gets the fourth item as well, and since at this point she has no edge towards agent 1, the mechanism stops.
        \item Agent 2 makes a declaration where \newx{$b_{21}=0$}. In that case, agent 1 declares \newx{$b_{11}=0$, $b_{12}=0$, $b_{13}=1$, and $b_{14}=1$}. Then the first two items go to agent 1\newx{, as nobody has an edge after the first item}. At this point, if agent 2 does not have an edge towards agent 1, then agent 1 gets the third item  as well, and she ends up with two items from \newx{the set $\{2,3,4\}$} in the final allocation. If on the other hand agent 2 does have an edge towards agent 1, then agent 2 gets the third item. In case  agent 2 removes her edge, then agent 1 gets items 1,2 and 4. In case agent 2 keeps her edge, notice that agent 1 also has an edge towards agent 2, so the cycle between the two agents will be eliminated and agent 1 ends up getting items 3, 4.
        \item Agent 2 makes a declaration where \newx{$b_{21}> 0$}. In that case, agent 1 declares \newx{$b_{11}=0$, $b_{12}=1$, $b_{13}=0$, and $b_{14}=1$}. Then the first item goes to the first agent and the second item goes to the second agent. At this point, agent 1 has an edge towards agent 2. If agent 2 does have an edge towards agent 1 too, the cycle between them is eliminated, agent 1 gets item 2 and at the next step she gets item 3 as well. Therefore, she ends up with at least two items from $\{2,3,4\}$. If on the other hand agent 2 does not have an edge towards agent 1, then agent 1 gets item 3 as well. If agent 2 adds an edge towards agent 1 at this point, since agent 1 also has an edge towards agent 2 (recall that at this step she has items 1 and 3 that both have 0 value for her according to what she declared), the cycle between the two is eliminated (therefore she gets item 2), and in the next step agent 1 gets item 4 as well. Therefore, she ends up with items 2 and 4. Finally, if agent 2 does not add an edge towards agent 1 at this point, then agent 1 gets item 4 as well, and thus, items 1,3, and 4 in the final allocation.
    \end{itemize}

    As these four cases capture every possible scenario, we can derive that there is always a profitable deviation for an agent regardless of what the other agent declared, and thus no pure Nash equilibrium exists for this instance.\end{proof}

  \begin{corollary}\label{cor:idbin}
    There are instances where Priority-E-C-E does not have pure Nash equilibria. These instances can be as simple as having just 2 agents with identical  binary valuation functions and 4 items.
\end{corollary}

\begin{proof}
    As with Corollary \ref{cor:2una}, the statement is derived directly from the fact that Standard-E-C-E and Priority-E-C-E run identically for instances with only two agents. \end{proof}

\newx{The instances of Theorems~\ref{thm:noPNE-Best} and~\ref{thm:noPNE-Standard} rely on the way items are chosen in the corresponding variants. Indeed, both instances do have pure Nash equilibria under Preferred-Item-E-C-E.}

\begin{revisedx}
\begin{proposition}\label{prop:pref-pne}
For the instances of Theorems~\ref{thm:noPNE-Best} and~\ref{thm:noPNE-Standard}, Preferred-Item-E-C-E has a pure Nash equilibrium that induces an EF1 allocation according to the true valuation functions.
\end{proposition}

\begin{proof}
\emph{The instance of Theorem~\ref{thm:noPNE-Best}.} Both agents report values $(0,1,1)$ and the preference list $(g_1,g_2,g_3)$. Agent $1$ receives $g_1$ in the first step. Agent $2$ has no edge, as she reports value $0$ for $g_1$, so agent $1$ also receives $g_2$. Now agent~$2$ has an edge towards agent $1$ and receives $g_3$; then both agents report value $1$ for both bundles, so no edge remains. The allocation $(\{g_1,g_2\},\{g_3\})$ is EF1.

Agent $1$ cannot get all three items. Indeed, agent $2$ reports a positive value for $g_2$ and $g_3$, so she has an edge towards agent $1$ as soon as agent $1$ holds one of them while agent $2$ holds nothing. If this happens after the first step, agent $2$ receives the second item; otherwise agent $1$ holds $g_1$ only after the first step, receives some item of $\{g_2,g_3\}$ in the second step, and then agent $2$ receives the third item. Once agent $2$ holds an item, she holds a nonempty bundle until the end, as cycle resolutions only exchange nonempty bundles by Observation~\ref{obs:basic}(d).

Agent $2$ cannot get two items. Agent $1$ receives $g_1$ in the first step, whatever agent $2$ reports. If agent $2$ then has no edge towards agent $1$, agent $1$ receives $g_2$, and agent $2$ receives at most the remaining item. Otherwise agent $2$ receives the second item $y\in\{g_2,g_3\}$, and agent $1$ adds an edge towards her, as she reports value $0$ for $g_1$ and $1$ for $y$. If agent $2$ still has an edge towards agent $1$, the bundles are swapped and afterwards no agent has an edge; if not, agent $1$ is the unique source. In both cases agent $1$ receives the last item, and agent $2$ ends up with one item. \smallskip

\emph{The instance of Theorem~\ref{thm:noPNE-Standard}.} Both agents report values $(0,0,1,1)$ and the preference list $(g_2,g_3,g_4,g_1)$. Agent $1$ receives $g_2$ and, as nobody has an edge, also $g_3$. Then agent $2$ has an edge towards agent $1$ and receives $g_4$, after which both agents report value $1$ for both bundles. Finally, agent $1$ receives $g_1$. The allocation $(\{g_1,g_2,g_3\},\{g_4\})$ gives value $2$ to agent $1$ and value $1$ to agent $2$, and it is EF1.

Agent $1$ cannot get all of $g_2,g_3,g_4$. First note that whenever agent $2$ is selected while one of $g_2,g_3,g_4$ is available, she receives such an item, by her preference list. Moreover, once agent $2$ holds one of these items, she holds one until the end: in a cycle resolution she receives a bundle for which she reports a larger value, i.e., one containing $g_3$ or $g_4$. Suppose that agent $1$ ends up with all three items. Then agent $2$ never holds any of them, so agent $1$ receives $g_3$ and $g_4$ in different steps, by her own picks. Right after the first of these steps, agent $2$ has an edge towards agent $1$, while one of $g_3,g_4$ is still available. If agent $1$ has an edge towards agent $2$ at that moment, the two bundles are swapped and agent $2$ holds $g_3$ or $g_4$. Otherwise agent $1$ is not a source, so agent $2$ is selected next while one of $g_3,g_4$ is available. Both cases lead to a contradiction.

Agent $2$ cannot get two of $g_2,g_3,g_4$. Agent $1$ receives $g_2$ in the first step, whatever agent $2$ reports. We consider the second step.
\begin{itemize}[itemsep=1pt, topsep=2pt, leftmargin=12pt]
\item If agent $2$ has no edge towards agent $1$ after the first step, agent $1$ receives $g_3$. From then on, agent $1$ reports value at least $1$ for her bundle and at most $1$ for any bundle of agent $2$ that does not contain both $g_3$ and $g_4$; so she has no edge and no cycle resolution happens, and agent $2$ may only receive items among $g_1$ and $g_4$.
\item If agent $2$ has an edge towards agent $1$ after the first step, she receives some item $y$. If $y\in\{g_3,g_4\}$, agent $1$ adds an edge towards agent $2$. If the bundles are swapped, agent $2$ holds $g_2$, agent $1$ holds $y$, and no agent has an edge; then agent $1$ receives the other item of $\{g_3,g_4\}$. From then on agent $1$ reports value $2$ for her bundle and $0$ for the bundle of agent $2$, so no cycle resolution happens, and agent $2$ may only receive $g_1$. If the bundles are not swapped, agent $1$ is the unique source and receives the other item $z$ of $\{g_3,g_4\}$; from then on agent $1$ reports value $1$ for both bundles, so no cycle resolution happens, and agent $2$ may only receive $g_1$. If $y=g_1$, agent $1$ has no edge. Either agent $2$ has an edge and receives one more item $y'\in\{g_3,g_4\}$, after which agent $1$ has an edge towards her: then either the bundles are swapped, after which no agent has an edge, agent $1$ receives the last item, and agent $2$ ends up with $\{g_2\}$; or agent $1$ is the unique source and receives the last item, and agent $2$ ends up with $\{g_1,y'\}$. Or agent $2$ has no edge, and agent $1$ receives $g_3$; then, as in the first case, agent $2$ may only receive $g_4$.
\end{itemize}
In every case, agent $2$ ends up with at most one of $g_2,g_3,g_4$.
\end{proof}
\end{revisedx}

Due to the results above, in the next section we turn our attention to what is achievable in terms of fairness guarantees, when agents play their best responses to what the rest of the agents have declared. As we will show, despite the complex nature of E-C-E, \textit{there are} versions where when an agent plays her best strategy  towards the others (and thus achieving the best possible utility for her), the fairness guarantees of the algorithm can be (approximately) preserved.

\section{E-C-E and Fair Best Responses: EF1}\label{sec:BR}

As we have so far shown,  the considered versions of E-C-E do not have good behavior in terms of the existence of pure Nash equilibria, even for very simple scenarios. This motivates our next approach, where we explore whether we can have fairness guarantees for an agent who plays her best response to what the rest of the agents played. We point out here that proving that any best response provides fairness guarantees  would mean that an allocation would provide that same guarantee to everyone at a pure Nash equilibrium (in case one exists), as a pure Nash equilibrium is a collection of best responses. Therefore, \textit{any fairness guarantee for best-responding agents also holds for all agents in any (potentially) existing pure Nash equilibrium}.

\begin{revisedx}
We first develop some simple tools that capture what a deviating agent can achieve when the choice of items does not depend on reported values and we use them for the case of two agents (Section~\ref{sec:BR2}). Then we show that a single deviation, namely truth-telling, already gives a guarantee for any number of agents, unless several single items are very valuable for the deviating agent (Section~\ref{sec:BRn}). For three agents, such items can be handled when the selected agent chooses her item (Section~\ref{sec:BR3}), but for four or more agents they cannot (Section~\ref{sec:BR4}).

All results of this section hold both for additive and for monotone subadditive valuation functions, with a loss in the approximation factor in the latter case. Recall that, in the subadditive setting, the bid vector of an agent is $\bm{b}_i = (f_{i}; \textrm{pref}_i)$, where $f_{i}$ is a value oracle for the monotone subadditive function agent $i$ wants to declare. Throughout this section, $i$ is a fixed agent and the bids $\mathbf{b}_{-i}$ of the other agents are fixed as well.

\subsection{Tools for Deviations}\label{sec:tools}

Our first lemma says that, in a value-independent variant (i.e., in Standard-, Priority- and Preferred-Item-E-C-E), the reported values of agent $i$ for items that have not been allocated yet do not affect the run so far.

\begin{lemma}\label{lem:locality}
Consider a value-independent variant of E-C-E and two bids $\bm b_i,\bm b'_i$ of agent $i$ with the same preference list. Let $T\subseteq M$ be such that $\hv_i(S)=\hv'_i(S)$ for all $S\subseteq T$, and let $s$ be the first step of the run under $(\bm b_i,\mathbf{b}_{-i})$ in which an item not in $T$ is selected. Then the runs under $(\bm b_i,\mathbf{b}_{-i})$ and $(\bm b'_i,\mathbf{b}_{-i})$ coincide up to, and including, the addition of the item of step $s$ to the bundle of the selected agent (and they coincide entirely, if no such step exists).
\end{lemma}

\begin{proof}
A run is a sequence of operations: selections of a source, selections of an item, additions, and rotations. We show, by induction on the number of operations performed, that the two runs perform the same operations up to, and including, the addition of the item of step $s$ (or until the end, if there is no such step); in particular, they have the same bundles and the same available items at every moment up to that point. Both runs start with empty bundles and with all items available. Now consider an operation of the run under $(\bm b_i,\mathbf{b}_{-i})$ that is not later than the addition of the item of step $s$, and assume that the two runs coincide right before it. At that moment, all allocated items belong to $T$, so every bundle present is a subset of $T$. Hence, every comparison of agent $i$ in the envy graph gives the same answer under $\hv_i$ and $\hv'_i$, while the comparisons of the other agents do not depend on agent $i$'s bid at all. So the two envy graphs are identical, and so are the selected source and the resolved cycles, as these depend only on the envy graph. Moreover, in a value-independent variant the selected item depends only on the set of available items, the fixed ordering of the items and the preference list of the selected agent, all of which are the same in both runs. Therefore, both runs perform the same operation, and they coincide right after it. This completes the induction.
\end{proof}

The second lemma is our main tool. By reporting a huge value for a single item, agent $i$ can make sure that the bundle containing this item never grows again while another agent holds it, and that she keeps it forever once she gets it.

\begin{lemma}\label{lem:boost}
Consider a value-independent variant of E-C-E and a bid $\bm b_i$ of agent $i$. Suppose that, in the run under $(\bm b_i,\mathbf{b}_{-i})$, an item $k$ is added in some step to the bundle of an agent $a$, forming the bundle $Z$. Let $\bm b'_i$ have the same preference list as $\bm b_i$ and report $\hv'_i(S)=\hv_i(S - k)+c\cdot\mathbf{1}[k\in S]$ for all $S\subseteq M$, where $c>\hv_i(M)$. Then $\hv'_i$ is additive if $\hv_i$ is, and monotone subadditive if $\hv_i$ is. Moreover, the runs under $(\bm b_i,\mathbf{b}_{-i})$ and $(\bm b'_i,\mathbf{b}_{-i})$ coincide up to, and including, the addition of $k$, and in the run under $(\bm b'_i,\mathbf{b}_{-i})$ the following hold.
\begin{enumerate}[label=(\roman*),itemsep=0pt,topsep=2pt]
\item If $a=i$, then the final bundle of agent $i$ contains $Z$.
\item If $a\neq i$, then, as long as agent $i$ does not hold $k$, she has an edge towards the holder of $k$, and the bundle containing $k$ is exactly $Z$. If agent $i$ ever holds $Z$, then her final bundle contains $Z$.
\end{enumerate}
\end{lemma}

\begin{proof}
For additive $\hv_i$, the function $\hv'_i$ is additive and just sets the value of $k$ to $c$. For monotone subadditive $\hv_i$, both $S\mapsto \hv_i(S-k)$ and $S\mapsto c\cdot \mathbf{1}[k\in S]$ are monotone and subadditive, and so is their sum. The two functions agree on all subsets of $M-k$, so, by Lemma~\ref{lem:locality}, the two runs coincide up to, and including, the addition of $k$.

After the addition of $k$, every bundle containing $k$ has reported value at least $c$ for agent $i$, while every bundle not containing $k$ has reported value at most $\hv_i(M)<c$. If agent $i$ holds $k$, she has no edge, so she belongs to no cycle and her bundle can only grow until the end. This proves (i), and the second part of (ii). If agent $i$ does not hold $k$, she has an edge towards the holder of $k$, so, by Observation~\ref{obs:basic}(b), no item is added to the bundle containing $k$. Since cycle resolutions do not change bundles, this bundle is $Z$ as long as agent $i$ does not hold $k$.
\end{proof}

The next lemma is about a deviation that is available in every variant: reporting truthfully. It is a reformulation of the EF1 guarantee of Observation~\ref{obs:basic}(e).

\begin{lemma}\label{lem:witness}
Consider any variant of E-C-E, and let $v_i$ be monotone subadditive. Suppose that agent $i$ reports truthfully (with an arbitrary preference list), and let $v_i(A_i)$ be her final value. Then there is a set $T$ of at most $n-1$ items, at most one from each final bundle of the other agents, such that $v_i(M\setminus T)\le n\cdot v_i(A_i)$.
\end{lemma}

\begin{proof}
By Observation~\ref{obs:basic}(e), every nonempty final bundle $P$ of an agent other than $i$ contains an item $g_P$ such that $v_i(P- g_P)\le v_i(A_i)$. Let $T$ be the set of these items. Then $M\setminus T$ is the union of the final bundle of agent $i$ and the sets $P- g_P$, together with possibly some empty bundles. There are $n$ sets in total, each of value at most $v_i(A_i)$. Subadditivity implies $v_i(M\setminus T)\le n\cdot v_i(A_i)$.
\end{proof}
\end{revisedx}

\subsection{Instances with Two Agents}\label{sec:BR2}

We begin with the case of two agents. As we will see in the next theorem, most of the considered versions of E-C-E guarantee EF1 fairness, if we look at the allocation from the perspective of an agent that plays her best response.

\begin{revisedx}
\begin{theorem}\label{thm:2EF1}
Consider any version of E-C-E among Standard, Priority, Preferred-Item, and an instance with $2$ agents. Then, the produced allocation is always EF1 from the perspective of any agent who plays a best response and has an additive valuation function.
\end{theorem}

\begin{proof}
Let $\bm b_i$ be a best response of agent $i$ to the bid of the other agent, say agent $j$, and let $(A_1,A_2)$ be the produced allocation. If $v_i(A_j)\le v_i(A_i)$, there is nothing to show. Otherwise, during the run some bundle has value more than $v_i(A_i)$ for agent $i$; let $t$ be the first step during which this happens. Since a cycle resolution only permutes the bundles present, this happens because of the addition of the item of step $t$, say $k$, to the bundle of some agent $a$, forming a bundle $B$ with $v_i(B)>v_i(A_i)$. Moreover, every bundle present before this addition, including $B- k$, has value at most $v_i(A_i)$ for agent $i$.

If $a=i$, Lemma~\ref{lem:boost}(i) gives a bid for agent $i$ under which her final bundle contains $B$, and so it has value more than $v_i(A_i)$. This contradicts the fact that $\bm b_i$ is a best response. Thus, $a=j$.

Now let agent $i$ deviate to the bid $\bm b'_i$ of Lemma~\ref{lem:boost} (for this $k$). By Lemma~\ref{lem:boost}(ii), after the addition of $k$ and as long as agent $i$ does not hold $k$, agent $i$ has an edge towards agent $j$, who holds $B$. If at some moment agent $j$ also has an edge towards agent $i$, the two bundles are swapped, and agent $i$ holds $B$. By Lemma~\ref{lem:boost}(ii), her final bundle then contains $B$, and its value is more than $v_i(A_i)$, a contradiction. Otherwise, at every subsequent selection agent $j$ has an incoming edge and agent $i$ does not, so agent $i$ is selected and receives every remaining item. Her final bundle is $M\setminus B$, and since $\bm b_i$ is a best response,
\[
v_i(M\setminus B)\le v_i(A_i).
\]
Now we return to the original run. By Observation~\ref{obs:basic}(a), $B$ is contained in a final bundle. This bundle is not $A_i$, since $v_i(B)>v_i(A_i)$ and $v_i$ is monotone. So $B\subseteq A_j$, and $A_i$ and $A_j\setminus B$ are disjoint subsets of $M\setminus B$.

When $v_i$ is additive, $v_i(A_j\setminus B)\le v_i(M\setminus B)-v_i(A_i)\le 0$ and, thus, $v_i(A_j\setminus B) = 0$. Therefore,
\[
v_i(A_j- k)=v_i(B- k)+v_i(A_j\setminus B) =v_i(B- k) \le v_i(A_i),
\]
i.e., the allocation is EF1 from agent $i$'s perspective.

\end{proof}

Best-Item-E-C-E is not captured by Theorem \ref{thm:2EF1}. The reason is that in Best-Item-E-C-E raising the reported value of an item may change the items agent $i$ picks before that item is allocated, as she would pick it at her first opportunity. So, Lemma~\ref{lem:boost} does not apply. Nevertheless, truth-telling alone gives a weaker guarantee for every variant.

\begin{proposition}\label{prop:two-best}
Consider any version of E-C-E and an instance with $2$ agents. Then, the produced allocation is always $\frac12$-EF1 from the perspective of any agent who plays a best response and has a monotone subadditive valuation function.
\end{proposition}

\begin{proof}
Let $\bm b_i$ be a best response of agent $i$, let $(A_1,A_2)$ be the produced allocation, and let $j$ be the other agent. If agent $i$ reported truthfully instead, her final value $w$ would be at most $v_i(A_i)$. Let $T$ be the set of Lemma~\ref{lem:witness} for this truthful run; then $|T|\le1$ and $v_i(M\setminus T)\le 2\, w\le 2\, v_i(A_i)$. Suppose that $A_j\neq\emptyset$. If $A_j$ contains the item of $T$, let $g$ be this item; otherwise let $g$ be any item of $A_j$. In both cases $A_j- g\subseteq M\setminus T$, so $v_i(A_j- g)\le 2\, v_i(A_i)$.
\end{proof}
\end{revisedx}

\begin{revisedx}
\subsection{Any Number of Agents: At Most One Large Item}\label{sec:BRn}

The argument of Proposition~\ref{prop:two-best} extends to any number of agents, as long as at most one single item is worth more to agent $i$ than her final bundle. As we will see in Section~\ref{sec:BR4}, some condition of this kind is necessary.

\begin{proposition}\label{prop:witness}
Consider any version of E-C-E and an instance with $n\ge3$ agents. Let agent $i$ play a best response, let $(A_1,\ldots,A_n)$ be the produced allocation, and suppose that there is at most one item $g\in M$ with $v_i(g)>v_i(A_i)$. Then the allocation is $\frac{1}{2n-3}$-EF1 from agent $i$'s perspective if $v_i$ is additive, and $\frac{1}{2n-2}$-EF1 from agent $i$'s perspective if $v_i$ is monotone subadditive.
\end{proposition}

\begin{proof}
If agent $i$ reports truthfully instead, her final value $w$ is at most $v_i(A_i)$. Let $T$ be the set of Lemma~\ref{lem:witness} for this truthful run; then $|T|\le n-1$ and $v_i(M\setminus T)\le n\cdot v_i(A_i)$. Fix an agent $j\neq i$ with $A_j\neq\emptyset$, and note that $|A_j\cap T| + |A_i\cap T|\le n-1$.
If $|A_j\cap T|\ge1$, let $g$ be an item of $A_j\cap T$ of maximum value for agent $i$; otherwise let $g$ be any item of $A_j$. Then
\[
A_j- g\subseteq (A_j\setminus T)\cup\big((A_j\cap T)- g\big).
\]
The set $(A_j\cap T)- g$ contains at most $\max\{|A_j\cap T|-1,0\}\le n-2$ items, each of value at most $v_i(A_i)$ for agent $i$. Indeed, if the unique item of value more than $v_i(A_i)$ exists and belongs to $A_j\cap T$, then it is $g$ itself, by the choice of $g$.

When $v_i$ is monotone subadditive, by subadditivity and Lemma~\ref{lem:witness},  we get
\[
v_i(A_j- g)\le v_i(M\setminus T)+(n-2)\, v_i(A_i)\le (2n-2)\, v_i(A_i).
\]

When $v_i$ is additive, we use that $A_j\setminus T$ and $A_i\setminus T$ are disjoint subsets of $M\setminus T$, so that $v_i(A_j\setminus T)\le v_i(M\setminus T) -v_i(A_i\setminus T) \le n\, v_i(A_i)-v_i(A_i\setminus T)$, by Lemma~\ref{lem:witness}. We distinguish three cases.
\begin{itemize}[itemsep=1pt, topsep=2pt, leftmargin=12pt]
\item If $|A_j\cap T|=0$, then $v_i(A_j- g)\le v_i(A_j)\le n\, v_i(A_i)\le (2n-3)\, v_i(A_i)$, where the last inequality uses $n\ge3$.
\item If $|A_j\cap T|\ge1$ and $|A_i\cap T|=0$, then $v_i(A_i\setminus T)=v_i(A_i)$, so $v_i(A_j- g)\le (n-1)\, v_i(A_i)+(|A_j\cap T|-1)\, v_i(A_i)\le(2n-3)\, v_i(A_i)$, where the last inequality uses $|A_j\cap T|\le n-1$.
\item If $|A_j\cap T|\ge1$ and $|A_i\cap T|\ge1$, then $|A_j\cap T|\le n-2$, so $v_i(A_j- g)\le n\, v_i(A_i)+(|A_j\cap T|-1)\, v_i(A_i)\le(2n-3)\, v_i(A_i)$.
\end{itemize}
In all cases the allocation is $\frac{1}{2n-3}$-EF1 from agent $i$'s perspective.
\end{proof}

In particular, Proposition~\ref{prop:witness} applies when every single item is worth at most $v_i(A_i)$ to agent $i$. Note that, by Corollary~\ref{col:trmms} in Section~\ref{sec:BR-MMS}, this is the case, e.g., under Best-Item-E-C-E and Preferred-Item-E-C-E when $v_i$ is additive and every item is worth at most $\frac{n}{2n-1}$ of agent $i$'s maximin share.
\end{revisedx}

\subsection{Instances with Three Agents}\label{sec:BR3}
In this section, we go beyond the case of instances with just two agents. As expected, adding just one more agent can make the strategy space a lot more complex. A result of this, as the next theorem demonstrates, is that for certain versions of E-C-E, the fairness properties of best responses (according to the true valuation functions of the agents) are not necessarily preserved anymore. Surprisingly, this can be shown even for instances where the true valuation functions of the agents are binary, and there are only two items to be allocated!

\begin{theorem}\label{thm:standard-noEF1-PNE}
    Consider Standard-E-C-E and instances with three additive binary agents. Then, there exist instances where some best responses induce allocations that are not (even approximately) EF1, with respect to the true valuation functions of the agents.
\end{theorem}

\begin{proof}
Consider an instance with three agents and two items. Now, let agent 1 and agent 2 have a value of 0 for both the items, while agent 3 has a value of 1, again for both of them. Regarding the bidding profile, let agents 1 and 2 report truthfully (therefore, they never add an edge towards another agent at any step of the algorithm and no cycle is ever formed). Notice that they are indifferent between any two allocations. We are, thus, interested in agent 3. We will show that no matter her declaration, she cannot get a positive value.

So, we  begin with agent 1 who gets the first item, as there is no edge in the graph and ties between sources are broken lexicographically. Now agent 3, based on what she declares, either has an edge towards agent 1 or not. In case her declaration is such that she has an edge, then agent 2 gets the second item (as both agents 2 and 3 are sources, and the tie between them is resolved lexicographically) and agent 3 ends up with a value of 0. In case she does not have an edge towards agent 1, then agent 1 gets the second item (as all the agents are sources, and the tie between them is resolved lexicographically) and agent 3, once again, ends up with a value of 0. Since these are the only two possibilities, and in both of them she gets a value of 0, both of them are best responses. However, the second one is not $\epsilon$-EF1 for her, as she envies agent 1 up to one item (for any $ \epsilon\in (0, 1]$). \end{proof}

If one looks into the proof of this negative result, they would see that it is quite brittle, depending on a specific allocation order for the items and on having value ties.
So, it makes sense to look for complementing positive results next. Indeed, we  show that EF1 can be approximately guaranteed for agents who play best responses in instances with three agents. As we have already mentioned, the presence of a third agent makes the analysis of the behavior of the agents (and the behavior itself) much more complex. Therefore, we are able to prove the next statement only for \newx{Best-Item-E-C-E and Preferred-Item-E-C-E, where the selected agent chooses her item.}

\begin{revisedx}
By Proposition~\ref{prop:witness}, the only obstacle is the presence of two or more items that are each worth more to agent $i$ than her final bundle. The next two lemmas show that, when the selected agent chooses her item, such items are all allocated at the very beginning of every run, independently of agent $i$'s bid. Both lemmas use that, in Best-Item-E-C-E and Preferred-Item-E-C-E, the sources are selected as in Priority-E-C-E.

\begin{lemma}\label{lem:access}
Consider Best-Item-E-C-E or Preferred-Item-E-C-E with any number of agents. There is a set $T_i$ of at most $i-1$ items such that, for every item $h\notin T_i$, agent $i$ has an additive bid under which her final bundle contains $h$. In particular, if no bid of agent $i$ gives her value more than some $x\ge0$, then every item $h$ with $v_i(h)>x$ belongs to $T_i$.
\end{lemma}

\begin{proof}
Let $\bm\beta$ be a bid of agent $i$ that reports value $1$ for every item, with an arbitrary preference list, and consider the run under $(\bm\beta,\mathbf{b}_{-i})$. As long as the bundle of agent $i$ is empty, she has an edge towards every agent with a nonempty bundle, and, by Observation~\ref{obs:basic}(d), she belongs to no cycle. So, by Observation~\ref{obs:basic}(b), until agent $i$ is selected for the first time, every item is added to an empty bundle. Moreover, the set of agents with nonempty bundles does not change in cycle resolutions: every agent on a cycle receives a bundle towards whose holder she had an edge, and such a bundle is nonempty by Observation~\ref{obs:basic}(d). In the first step no agent has an edge, so agent $1$ is selected. From the second step on, agent $i$ is a source with an outgoing edge, so the selected agent is either agent $i$ or a source with an outgoing edge and a smaller index than $i$. We conclude that, until agent $i$ is selected for the first time, the items are added to the bundles of pairwise distinct agents with indices smaller than $i$. Let $T_i$ be the set of items allocated before agent $i$ is selected for the first time ($T_i$ can be all of $M$, if she is never selected); then $|T_i|\le i-1$.

Now let $h\notin T_i$, and let $\bm\beta^h$ be the bid of agent $i$ that reports value $m+1$ for $h$ and value $1$ for every other item, and has $h$ first in its preference list. Under both $\bm\beta$ and $\bm\beta^h$, as long as agent $i$ has an empty bundle, she has an edge towards exactly the agents with nonempty bundles, while the edges of the other agents and the items they choose depend only on their own bids and on the current bundles. By induction, the runs under $(\bm\beta,\mathbf{b}_{-i})$ and $(\bm\beta^h,\mathbf{b}_{-i})$ coincide until agent $i$ is selected for the first time, and at that moment $h$ is still available. Then agent $i$ receives $h$: under Preferred-Item-E-C-E it is first in her preference list, and under Best-Item-E-C-E it is her unique item of maximum reported (marginal) value, as her bid is additive. From then on, her reported value for her own bundle is at least $m+1$, while every other bundle has reported value at most $m-1$. So, she never has an edge, she belongs to no cycle, and her bundle only grows until the end.

Note that the second part of the statement directly follows from the first. Indeed, if some item $h$ with $v_i(h)>x$ were not in $T_i$, then, by monotonicity, the bid $\bm\beta^h$ would give agent $i$ value at least $v_i(h)>x$.
\end{proof}

\begin{lemma}\label{lem:twohigh}
Consider Best-Item-E-C-E or Preferred-Item-E-C-E with $3$ agents, and suppose that no bid of agent $i$ gives her value more than some $x\ge0$. If two distinct items have value more than $x$ for agent $i$, then $i=3$, and there are distinct items $x_1, x_2$ such that, whatever agent $i$ bids, agent $1$ receives $x_1$ in the first step and agent $2$ receives $x_2$ in the second step. Moreover, $x_1$ and $x_2$ are the two items of value more than $x$, and $v_i(M\setminus\{x_1,x_2\})\le x$.
\end{lemma}

\begin{proof}
By Lemma~\ref{lem:access}, both items of value more than $x$ belong to $T_i$. So $|T_i|\ge2$, which, given that $|T_i|\le i-1$, implies $i=3$ and that $T_3$ consists of exactly these two items. Let $x_1$ and $x_2$ be the items allocated in the first and the second step of the run under $(\bm\beta,\mathbf{b}_{-i})$ of the proof of Lemma~\ref{lem:access}.

In every run, in the first step no agent has an edge, so agent $1$ is selected and receives $x_1$, which only depends on her own bid, as her bundle is empty when she is selected. In the run under $(\bm\beta,\mathbf{b}_{-i})$, agent $2$ is selected in the second step, although agent $3$ is a source with an outgoing edge at that moment. So agent $2$ has an edge towards agent $1$ after the first step, and this only depends on agent $2$'s bid. Therefore, in every run, after the first step agent $1$ has an incoming edge and agent $2$ is a source with an outgoing edge and a smaller index than agent $3$. So agent $2$ is selected in the second step, and she receives $x_2$, which only depends on her own bid, as her bundle is empty when she is selected. After the second step, $\{x_1\}$ and $\{x_2\}$ are held by agents $1$ and $2$, in some order, and agent $3$ holds nothing.

Now let agent $3$ bid an additive bid that reports value $1$ for $x_1$ and $x_2$, and value $0$ for every other item, with an arbitrary preference list.

\begin{claim}\label{cl:twohigh}
Under this bid, the final bundle of agent $3$ contains $x_1$ or $x_2$, or it is equal to $M\setminus\{x_1,x_2\}$.
\end{claim}

\begin{proof}[Proof of Claim~\ref{cl:twohigh}]
As long as agent $3$ holds neither $x_1$ nor $x_2$, her reported value for her own bundle is $0$, and she has an edge towards the holders of $\{x_1\}$ and $\{x_2\}$. By Observation~\ref{obs:basic}(b), these two bundles do not grow. When a source is selected the envy graph has no cycle, so no agent has an edge towards agent $3$, as otherwise this edge together with an edge of agent $3$ would form a cycle. Hence agent $3$ is the unique source and she receives the item. If a cycle resolution ever gives agent $3$ one of the bundles $\{x_1\}$ and $\{x_2\}$, her reported value for her own bundle becomes $1$, while every other bundle has reported value at most $1$. Then she never has an edge again, she belongs to no cycle, and her bundle only grows until the end. If no such cycle resolution happens, she receives every item after the second step.
\end{proof}

Both $x_1$ and $x_2$ have value more than $x$ for agent $3$, while it is assumed that no bid gives her value more than $x$. So, by Claim~\ref{cl:twohigh}, the final bundle of agent $3$ is $M\setminus\{x_1,x_2\}$, and $v_i(M\setminus\{x_1,x_2\})\le x$.
\end{proof}
\end{revisedx}

We are now ready to state and prove our main result for three agents.

\begin{revisedx}
\begin{theorem}\label{thm:3EF1}
Consider Best-Item-E-C-E or Preferred-Item-E-C-E, and an instance with $3$ agents. Then, the produced allocation from the perspective of any agent who plays a best response is always $\frac13$-EF1, if the agent has an additive valuation function, and $\frac14$-EF1, if the agent has a monotone subadditive valuation function.
\end{theorem}

\begin{proof}
Let agent $i$ play a best response, and let $(A_1,A_2,A_3)$ be the produced allocation. If there is at most one item $g\in M$ with $v_i(g)>v_i(A_i)$, then Proposition~\ref{prop:witness} for $n=3$ gives $\frac13$-EF1 in the additive case and $\frac14$-EF1 in the subadditive case.

Otherwise, we apply Lemma~\ref{lem:twohigh} for $x=v_i(A_i)$, and we get items $x_1,x_2$ with $v_i(M\setminus\{x_1,x_2\})\le v_i(A_i)$. In the original run, $x_1$ and $x_2$ belong to different bundles after the second step, so, by Observation~\ref{obs:basic}(a), no final bundle contains both. Moreover, neither of them belongs to $A_i$, as each is worth strictly more than $v_i(A_i)$ to agent $i$. Let $j\neq i$ with $A_j\neq\emptyset$. If $A_j$ contains $x_1$ or $x_2$, let $g$ be this item; otherwise let $g$ be any item of $A_j$. In both cases $A_j- g\subseteq M\setminus\{x_1,x_2\}$, so $v_i(A_j- g)\le v_i(A_i)$. That is, in this case the allocation is even EF1 from agent $i$'s perspective.
\end{proof}
\end{revisedx}

\subsection{Four or More Agents}\label{sec:BR4}
The positive results of Sections~\ref{sec:BR2} and~\ref{sec:BR3} do not extend to four or more agents. In fact, no variant we consider guarantees any approximation of EF1 to every best-responding agent, even at a pure Nash equilibrium.
 
\begin{theorem}\label{thm:four}
Let $n\ge4$ and consider any of Standard-, Priority-, Best-Item- and Preferred-Item-E-C-E. There is an instance with $n$ agents with additive valuation functions and $3$ items, together with a pure Nash equilibrium, such that the produced allocation is not $\alpha$-EF1 from the perspective of one of the agents, for any $\alpha>0$. 
\end{theorem}
 
\begin{proof}
Let the items be $p,q,r$, in this order. The true valuation functions and the bids of the agents are as follows, where agents $5,\ldots,n$ only exist if $n\ge5$.
\[
\begin{array}{c|ccc|ccc|c}
 & \multicolumn{3}{c|}{\text{true}} & \multicolumn{3}{c|}{\text{reported}} & \text{pref.~list}\\
\text{agent} & p&q&r & p&q&r & \\\hline
1 & 5&1&2 & 10&1&2 & p,q,r\\
2 & 1&5&2 & 4&2&1 & q,r,p\\
3 & 1&2&0 & 2&0&1 & p,r,q\\
4,\ldots,n & 1&2&3 & 1&0&2 & r,p,q
\end{array}
\]
Note that agent $1$ never has an edge while she holds $p$, as she reports value $10$ for $p$ and value $3$ for $\{q,r\}$.
 
We claim that agents $4,\ldots,n$ get nothing, whatever one of them bids. To see this, fix an agent $a\in\{4,\ldots,n\}$ and any bid of $a$, and keep the bids of the other agents fixed. In the first step no agent has an edge, so agent $1$ is selected, and she receives $p$ in every variant: $p$ is the first item in the fixed ordering of the items, it is first in her preference list, and it is her unique item of maximum reported value. Now agents $2$ and $3$ have an edge towards agent $1$, so agent $1$ is not a source. By Observation~\ref{obs:basic}(d), there is no cycle, and agents $2,3,\ldots,n$ are sources. Agent $2$ has the smallest index among them, and she has an outgoing edge; so she is selected in every variant, and she receives $q$ in every variant, as she reports value $2$ for $q$ and $1$ for $r$. After the second step, agent $2$ still has an edge towards agent $1$, as she reports value $2$ for $q$ and $4$ for $p$. Agent $1$ has no edges, agent $3$ reports value $0$ for $q$, and so does every agent of $\{4,\ldots,n\}$ other than $a$. Hence the only agent who may have an edge towards agent $2$ is agent $a$.
\begin{itemize}[itemsep=1pt, topsep=2pt, leftmargin=12pt]
\item If agent $a$ has no edge towards agent $2$, then agent $2$ is the source with the smallest index, and she has an outgoing edge. So she is selected, and she receives $r$.
\item If agent $a$ has an edge towards agent $2$, then agent $3$ is the source with the smallest index, and she has an outgoing edge, towards agent $1$. So she is selected, and she receives $r$.
\end{itemize}
Throughout the run, agent $1$ has no edges. Every other agent with a nonempty bundle only has an edge towards agent $1$: agent $2$ holds $\{q\}$ or $\{q,r\}$, for which she reports value $2$ or $3$, while she reports value $1$ for $\{r\}$; and if agent $3$ holds $\{r\}$, she reports value $1$ for it and value $0$ for $\{q\}$. Agents with empty bundles have no incoming edges, by Observation~\ref{obs:basic}(d). So a cycle could only consist of agents with nonempty bundles, and it would have to contain agent $1$, who has no edges. Hence the envy graph never contains a cycle, no cycle resolution happens, and agent $a$ ends up with an empty bundle.
 
Towards showing that that this reported profile is a pure Nash equilibrium, we first show that agents $2$ and $3$ cannot improve. Indeed, whatever agents $2,\ldots,n$ bid, agent $1$ is selected in the first step, receives $p$, and keeps it until the end, as she never has an edge while holding $p$. So agent $2$ cannot get more than $v_2(\{q,r\})$, which is what she gets here; see below. Now fix any bid of agent $3$. After the first step, agent $2$ has an edge towards agent $1$, and, as above, she is selected in the second step and receives $q$. Afterwards, agent $2$ only has edges towards agent $1$, since she reports value at most $1$ for every bundle not containing $p$ or $q$; so she belongs to no cycle and keeps $q$ until the end. Hence agent $3$ can only get $r$, which has value $0$ for her.
 
Moreover, we show that agent $1$ cannot improve either. It suffices to show that, whatever agent $1$ bids, her final bundle does not contain $p$ together with another item, as $v_1(p)=5>v_1(\{q,r\})$. Fix any bid of agent $1$.  She is selected in the first step and receives some item $x$. Agent $2$ reports a positive value for every item, so she has an edge towards agent $1$; she is a source with an outgoing edge and the smallest index among the sources $2,\ldots,n$, so she is selected in the second step and receives some item $y$. After the cycle resolutions (if any) of the second step, $\{x\}$ and $\{y\}$ are held by agents $1$ and $2$, in some order, and all other agents have empty bundles.
\begin{itemize}[itemsep=1pt, topsep=2pt, leftmargin=12pt]
\item If $p\in\{x,y\}$, agent $3$ has an edge towards the holder of $\{p\}$, as she reports value $2$ for it. By Observation~\ref{obs:basic}(b), the third item is not added to $\{p\}$, so no bundle ever contains $p$ together with another item.
\item If $p\notin\{x,y\}$, then $\{x,y\}=\{q,r\}$ and $p$ is the third item. Agent $3$ has an edge towards the holder of $\{r\}$, and if agent $2$ holds $\{r\}$, she has an edge towards the holder of $\{q\}$, i.e., towards agent $1$. In any case, agent $1$ is not a source, so $p$ is added to the bundle of another agent. If it is added to $\{q\}$ or $\{r\}$, held by agent $2$, then agent $2$ reports value at least $5$ for her bundle and at most $2$ for every other bundle; so she has no edges, belongs to no cycle, and keeps $p$. Otherwise $p$ is added to the empty bundle of agent $3$, and the bundle containing $p$ is $\{p\}$ until the end.
\end{itemize}
 
About the resulting output, by the first part of the proof, under the bids of the table the allocation is $A_1=\{p\}$, $A_2=\{q,r\}$, and $A_3=A_4=\cdots=A_n=\emptyset$, as agent $4$ reports value $0$ for $q$. For every $g\in A_2$ and every $\alpha>0$, we have $v_4(A_4)=0<\alpha\cdot v_4(A_2- g)$, as $v_4(A_2- g)\ge 2$.
\end{proof}

It is worth noting that in the instance used in the proof of Theorem~\ref{thm:four} every agent has pairwise distinct true values and pairwise distinct reported values for the three items.
However, like Theorem~\ref{thm:standard-noEF1-PNE}, Theorem~\ref{thm:four} exploits an agent who is indifferent among all her bids. Indeed, agent $4$ has other best responses, namely those that add an edge towards agent $2$ after the second step, which lead to the allocation $(\{p\},\{q\},\{r\},\emptyset,\ldots,\emptyset)$ that is EF1 from her perspective. Thus, Theorem~\ref{thm:four} rules out guarantees for \emph{every} best response, but not guarantees for \emph{some} best response of each agent. 
 
The instance of Theorem~\ref{thm:four} does not rely on ties between values, but it does rely on reported values equal to $0$: agents $3,\ldots,n$ report value $0$ for $q$. Otherwise, after the second step, every agent with an empty bundle would have an edge towards agent $2$, and agent $3$ would receive $r$. The next proposition shows that, for Best-Item-E-C-E and Preferred-Item-E-C-E, this is unavoidable: if all other agents report positive values, every best response comes with a guarantee.
 
\begin{proposition}\label{prop:positive}
Consider Best-Item-E-C-E or Preferred-Item-E-C-E and an instance with $n\ge2$ agents, and suppose that every agent other than agent $i$ reports a positive value for every single item. Then, the produced allocation is $\frac{1}{2n-2}$-EF1 from the perspective of agent $i$ if she plays a best response and has an additive valuation function, and $\frac{1}{2n-1}$-EF1 if she plays a best response and has a monotone subadditive valuation function.
\end{proposition}
 
\begin{proof}
Let agent $i$ play a best response and let $(A_1,\ldots,A_n)$ be the produced allocation. We call an item $h$ \emph{large} if $v_i(h)>v_i(A_i)$. Let $\ell=\min\{i-1,m\}$.
 
We first show that there are items $x_1,\ldots,x_{\ell}$ such that, in every run, whatever agent $i$ bids, agent $k$ is selected in step $k$ and receives $x_k$, for $k=1,\ldots,\ell$. We proceed by induction on $k$. In the first step no agent has an edge, so agent $1$ is selected and receives an item $x_1$ that depends only on her own bid. Let $2\le k\le\ell$ and assume the statement for $1,\ldots,k-1$. Before step $k$, the agents $1,\ldots,k-1$ hold one item each (rotations only exchange nonempty bundles among them, by Observation~\ref{obs:basic}(d)), and the agents $k,\ldots,n$ hold empty bundles. Agent $k\neq i$ has an empty bundle and reports positive values, so she has an edge towards each of the agents $1,\ldots,k-1$, who are thus not sources. Hence the sources are the agents $k,\ldots,n$, and agent $k$ has the smallest index among them and an outgoing edge; so she is selected. Since her bundle is empty, the item she receives depends only on her own bid and on the set $M\setminus\{x_1,\ldots,x_{k-1}\}$ of available items; call it $x_k$.
 
Next, consider the bid $\bm\beta$ of agent $i$ from the proof of Lemma~\ref{lem:access}, i.e., $i$ reports value $1$ for every item, with an arbitrary preference list. Under $\bm\beta$, the first $\ell$ steps are as above. If $\ell<m$, then right before step $i$ the agents $1,\ldots,i-1$ hold nonempty bundles and are not sources, as agent $i$ has an edge towards each of them; agent $i$ is a source with an outgoing edge and the smallest index among the sources $i,\ldots,n$, so she is selected. Hence, the set $T_i$ of Lemma~\ref{lem:access} is $\{x_1,\ldots,x_{\ell}\}$, and, by the same lemma, every large item is one of $x_1,\ldots,x_{\ell}$. In the original run, these items are added to the bundles of distinct agents, and all these bundles are present after step $\ell$; so, by Observation~\ref{obs:basic}(a), they belong to distinct final bundles. Thus every final bundle contains at most one large item and, by definition, $A_i$ contains none.
 
In the remaining of the proof, we work like in the proof of Proposition \ref{prop:witness}. Let $T$ be the set of Lemma~\ref{lem:witness} for the run where agent $i$ reports truthfully; since her final value there is at most $v_i(A_i)$, we have $|T|\le n-1$ and $v_i(M\setminus T)\le n\cdot v_i(A_i)$. Fix $j\neq i$ with $A_j\neq\emptyset$. If $A_j$ contains a large item, let $g$ be this item; otherwise, let $g$ be an item of $A_j\cap T$ if this set is nonempty, and any item of $A_j$ if it is empty. Then every item of $A_j- g$ has value at most $v_i(A_i)$ for agent $i$, and
\[
A_j- g\subseteq (A_j\setminus T)\cup\big((A_j\cap T)- g\big).
\]
Note that $|A_j\cap T| + |A_i\cap T|\le n-1$. When $v_i$ is monotone subadditive, we get $v_i(A_j- g)\le v_i(M\setminus T)+|A_j\cap T|\cdot v_i(A_i)\le (2n-1)\cdot v_i(A_i)$. When $v_i$ is additive, $A_j\setminus T$ and $A_i\setminus T$ are disjoint subsets of $M\setminus T$, so $v_i(A_j\setminus T)\le n\cdot v_i(A_i)-v_i(A_i\setminus T)$. If $|A_i\cap T|=0$, then $v_i(A_i\setminus T)=v_i(A_i)$ and $v_i(A_j- g)\le (n-1)\cdot v_i(A_i)+|A_j\cap T|\cdot v_i(A_i)\le(2n-2)\cdot v_i(A_i)$. If $|A_i\cap T|\ge1$, then $|A_j\cap T|\le n-2$ and $v_i(A_j- g)\le n\cdot v_i(A_i)+(n-2)\cdot v_i(A_i)=(2n-2)\cdot v_i(A_i)$.
\end{proof}
 
We do not know whether an analog of Proposition~\ref{prop:positive} holds for Standard-E-C-E and Priority-E-C-E, where agent $i$ cannot choose her items.
 
\medskip
 
We conclude this section by noting that all our positive results are robust to approximately optimal deviations, which is relevant as best responses may be hard to compute. For $0<\beta\le1$, we say that $\bm b_i$ is a \emph{$\beta$-best response} to $\mathbf{b}_{-i}$ if no bid of agent $i$ gives her value more than $\frac{1}{\beta}\cdot v_i(\mathcal{M}^{(i)}(\mathbf{b}))$.
 
\begin{remark}\label{rem:approx}
All proofs in this section use the best-response assumption only through statements of the form ``no bid of agent $i$ gives her value more than $v_i(A_i)$''. Replacing $v_i(A_i)$ by $\frac{1}{\beta}\cdot v_i(A_i)$ in these statements, and in the definitions of the step $t$ and of large items, the same proofs give the following guarantees for $\beta$-best responses. 
Following the proof of Theorem~\ref{thm:2EF1}, we get $\frac{\beta}{2-\beta}$-EF1  in the additive case; Proposition~\ref{prop:two-best} gives $\frac\beta2$-EF1 in the subadditive case, whereas following the proof of Proposition~\ref{prop:witness}, we get $\frac{\beta}{2n-2-\beta}$-EF1 in the additive and $\frac{\beta}{2n-2}$-EF1 in the subadditive case; accordingly, Theorem~\ref{thm:3EF1} gives $\frac{\beta}{4-\beta}$-EF1 and $\frac{\beta}{4}$-EF1.
\end{remark}

For completeness, we state these guarantees explicitly and spell out their derivations. Below, $A_i$ denotes the bundle of agent $i$ under $(\bm b_i,\mathbf b_{-i})$.

\begin{proposition}\label{prop:approx}
	Let $0<\beta\le1$, and let agent $i$ play a $\beta$-best response.
	\begin{enumerate}[label=(\alph*),itemsep=0pt,topsep=2pt]
		\item Under Standard-, Priority- or Preferred-Item-E-C-E with $2$ agents, the allocation is $\frac{\beta}{2-\beta}$-EF1 from agent $i$'s perspective if $v_i$ is additive.
		\item Under any variant with $2$ agents, the allocation is $\frac\beta2$-EF1 from agent $i$'s perspective if $v_i$ is monotone subadditive.
		\item Under any variant with $n\ge3$ agents, if at most one item $g$ has $v_i(g)>\frac1\beta\cdot v_i(A_i)$, the allocation is $\frac{\beta}{2n-2-\beta}$-EF1 from agent $i$'s perspective if $v_i$ is additive, and $\frac{\beta}{2n-2}$-EF1 if $v_i$ is monotone subadditive.
		\item Under Best-Item- or Preferred-Item-E-C-E with $3$ agents, the allocation is $\frac{\beta}{4-\beta}$-EF1 from agent $i$'s perspective if $v_i$ is additive, and $\frac\beta4$-EF1 if $v_i$ is monotone subadditive.
	\end{enumerate}
	For $\beta=1$, these are exactly the guarantees of Theorems~\ref{thm:2EF1} and~\ref{thm:3EF1} and Propositions~\ref{prop:two-best} and~\ref{prop:witness}.
\end{proposition}

\begin{proof}
	As noted in Remark~\ref{rem:approx}, we follow the proofs of this section, using the weaker statement ``no bid of agent $i$ gives her value more than $\frac1\beta\cdot v_i(A_i)$'' instead of ``no bid of agent $i$ gives her value more than $v_i(A_i)$''.
	
	(a) Let $j$ be the other agent. If $v_i(A_j)\le\frac1\beta\cdot v_i(A_i)$, the allocation is even $\beta$-EF1 from agent $i$'s perspective, and $\beta\ge\frac{\beta}{2-\beta}$. Otherwise, as in the proof of Theorem~\ref{thm:2EF1}, let $t$ be the first step during which some bundle has value more than $\frac1\beta\cdot v_i(A_i)$ for agent $i$; this is due to the addition of an item $k$ to the bundle of some agent $a$, forming a bundle $B$, and $v_i(B-k)\le\frac1\beta\cdot v_i(A_i)<v_i(B)$. Exactly as in that proof, Lemma~\ref{lem:boost} implies that $a=j$ and that $v_i(M\setminus B)\le\frac1\beta\cdot v_i(A_i)$. Moreover, $B\subseteq A_j$, as $v_i(B)>\frac1\beta\cdot v_i(A_i)\ge v_i(A_i)$. Since $v_i$ is additive, $v_i(A_j\setminus B)\le v_i(M\setminus B)-v_i(A_i)\le\frac1\beta\cdot v_i(A_i)-v_i(A_i)$, and so
	\[
	v_i(A_j-k)=v_i(B-k)+v_i(A_j\setminus B)\le\Big(\frac2\beta-1\Big)\cdot v_i(A_i)=\frac{2-\beta}{\beta}\cdot v_i(A_i).
	\]
	
	(b) As in the proof of Proposition~\ref{prop:two-best}, the truthful bid gives agent $i$ a value $w\le\frac1\beta\cdot v_i(A_i)$, so, for the set $T$ of Lemma~\ref{lem:witness}, $v_i(M\setminus T)\le 2w\le\frac2\beta\cdot v_i(A_i)$, and the same choice of $g$ gives $v_i(A_j-g)\le\frac2\beta\cdot v_i(A_i)$.
	
	(c) We follow the proof of Proposition~\ref{prop:witness}. Now the truthful bid gives agent $i$ a value $w\le\frac1\beta\cdot v_i(A_i)$, so $v_i(M\setminus T)\le n\,w\le\frac n\beta\cdot v_i(A_i)$, and each of the at most $n-2$ items of $(A_j\cap T)-g$ has value at most $\frac1\beta\cdot v_i(A_i)$. In the subadditive case, $v_i(A_j-g)\le\frac{n}{\beta}\cdot v_i(A_i)+\frac{n-2}{\beta}\cdot v_i(A_i)=\frac{2n-2}{\beta}\cdot v_i(A_i)$. In the additive case, the three cases of that proof give the following upper bounds on $v_i(A_j-g)$:
	\begin{itemize}[itemsep=1pt, topsep=2pt, leftmargin=12pt]
		\item if $|A_j\cap T|=0$: $\frac{n}{\beta}\cdot v_i(A_i)$;
		\item if $|A_j\cap T|\ge1$ and $|A_i\cap T|=0$: $\frac{n}{\beta}\cdot v_i(A_i)-v_i(A_i)+\frac{n-2}{\beta}\cdot v_i(A_i)=\big(\frac{2n-2}{\beta}-1\big)\cdot v_i(A_i)$;
		\item if $|A_j\cap T|\ge1$ and $|A_i\cap T|\ge1$: $\frac{n}{\beta}\cdot v_i(A_i)+\frac{n-3}{\beta}\cdot v_i(A_i)=\frac{2n-3}{\beta}\cdot v_i(A_i)$.
	\end{itemize}
	As $\beta\le1$ and $n\ge3$, the second bound is the largest, and it equals $\frac{2n-2-\beta}{\beta}\cdot v_i(A_i)$.
	
	(d) If at most one item is worth more than $\frac1\beta\cdot v_i(A_i)$ to agent $i$, part (c) for $n=3$ gives $\frac{\beta}{4-\beta}$-EF1 and $\frac{\beta}{4}$-EF1. Otherwise, Lemma~\ref{lem:twohigh} for $x=\frac1\beta\cdot v_i(A_i)$ gives items $x_1,x_2$ with $v_i(M\setminus\{x_1,x_2\})\le\frac1\beta\cdot v_i(A_i)$, neither of which belongs to $A_i$, as each is worth more than $\frac1\beta\cdot v_i(A_i)\ge v_i(A_i)$. As in the proof of Theorem~\ref{thm:3EF1}, the allocation is $\beta$-EF1 from agent $i$'s perspective, and $\beta\ge\frac{\beta}{4-\beta}\ge\frac\beta4$.
\end{proof}

\section{E-C-E and Fair Best Responses: MMS}\label{sec:BR-MMS}

\begin{revisedx}
Theorem~\ref{thm:four} shows that, for four or more agents, best responses may lead to allocations that are arbitrarily far from EF1. In this section we show that a weaker, share-based guarantee survives for any number of agents. For a valuation function $v_i$, an integer $n'\ge1$ and a set of items $M'\subseteq M$, the \emph{maximin share} of agent $i$ (with respect to $n',M'$) is
\[
\mu_i(n',M')=\max_{(P_1,\ldots,P_{n'})}\ \min_{r\in[n']} v_i(P_r),
\]
where the maximum is over all partitions of $M'$ into $n'$ (possibly empty) sets. An allocation is $\alpha$-MMS from agent $i$'s perspective if $v_i(A_i)\ge\alpha\cdot\mu_i(n,M)$.

The following monotonicity property of the maximin share is well known; for additive valuation functions it is due to Bouveret and Lema{\^\i}tre \cite[Lemma~1]{BouveretL16}. 
We include its short proof for completeness.

\begin{lemma}\label{lem:trmms}
Let $v_i$ be monotone, and let $T\subseteq M$ with $|T|<n$. Then $\mu_i(n-|T|,M\setminus T)\ge\mu_i(n,M)$. If, in addition, $v_i$ is additive, then $v_i(M\setminus T)\ge(n-|T|)\cdot\mu_i(n,M)$.
\end{lemma}

\begin{proof}
Let $(P_1,\ldots,P_n)$ be a partition of $M$ with $v_i(P_r)\ge\mu_i(n,M)$ for all $r\in[n]$. At most $|T|$ of its sets intersect $T$, so at least $n-|T|$ of them do not. Keep $n-|T|$ of these sets, and add every other item of $M\setminus T$ to one of them arbitrarily. This gives a partition of $M\setminus T$ into $n-|T|$ sets, each of which contains one of the original sets and, by monotonicity, has value at least $\mu_i(n,M)$. This proves the first inequality. When $v_i$ is additive, the value of $M\setminus T$ is the sum of the values of these $n-|T|$ sets, which proves the second inequality.
\end{proof}
\end{revisedx}

\begin{theorem}\label{thm:trmms}
\newx{Consider Best-Item-E-C-E and an agent $i$ with an additive valuation function who reports truthfully. Then, whatever the other agents report, $v_i(A_i)\ge\frac{n}{2n-1}\cdot\mu_i(n,M)$. The same holds under Preferred-Item-E-C-E, if agent $i$ reports truthfully and her preference list orders the items in non-increasing order of $v_i$. Moreover, the factor $\frac{n}{2n-1}$ is tight, even when all agents report truthfully.}
\end{theorem}

\begin{revisedx}
\begin{proof}
In both cases, agent $i$ has an edge towards another agent exactly when she truly envies her, and, whenever she is selected, she receives an item of maximum value for her among the available ones; under Best-Item-E-C-E this uses that, as her report is additive, her reported marginal value for an item is just its value. By Observation~\ref{obs:basic}(c), the value of agent $i$ for her own bundle never decreases, so it is at most $v_i(A_i)$ at every moment. By Observation~\ref{obs:basic}(d), agent $i$ belongs to no cycle while her bundle is empty, so her bundle is empty until she is selected for the first time. If she is selected at least once, let $g^*$ be the item she receives then. Right after receiving $g^*$, her value for her own bundle is at least $v_i(g^*)$, so $v_i(g^*)\le v_i(A_i)$.

\begin{claim}\label{cl:mms}
If, for some $j\ne i$, the bundle $A_j$ contains at least two items of positive value for agent $i$, then $v_i(A_j)\le 2\, v_i(A_i)$.
\end{claim}

\begin{proof}[Proof of Claim~\ref{cl:mms}]
Let $h$ be the item of $A_j$ of positive value for agent $i$ that is allocated last, say in step $s$, and let $Z$ be the bundle to which $h$ is added; $Z$ is held by the agent $b$ selected in step $s$. 
Every other item of $A_j$ of positive value for agent $i$ lies in a nonempty bundle that was present at the moment right after step $s$; by Observation~\ref{obs:basic}(a), this bundle is $Z\cup\{h\}$ and it is contained in $A_j$ in the final allocation. So $Z$ contains an item of positive value for agent $i$, and every item added to the bundle containing $Z\cup\{h\}$ after step $s$ has value $0$ for her (in particular, the set $A_j\setminus(Z\cup\{h\})$ has no value for agent $i$).
\begin{itemize}[itemsep=0pt,topsep=1pt, leftmargin=12pt]
\item If $b=i$, then $v_i(Z)\le v_i(A_i)$, as $Z$ is agent $i$'s own bundle at that moment. Moreover, $h$ is received by agent $i$ at or after her first selection, so $v_i(h)\le v_i(g^*)$.
\item If $b\ne i$, then agent $i$ has no edge towards agent $b$ when $b$ is selected, so $v_i(Z)$ is at most the value of agent $i$ for her own bundle, which is at most $v_i(A_i)$. Moreover, agent $i$ has been selected before step $s$: otherwise her bundle would be empty, and she would have an edge towards agent $b$, as $v_i(Z)>0$. (In particular, $g^*$ is well defined whenever the claim is not vacuous.) So, $h$ was available when agent $i$ received $g^*$, and $v_i(h)\le v_i(g^*)$.
\end{itemize}
In both cases $v_i(h)\le v_i(g^*)\le v_i(A_i)$. Since every item of $A_j\setminus(Z\cup\{h\})$ has value $0$ for agent $i$, we get $v_i(A_j)=v_i(Z)+v_i(h)\le 2\, v_i(A_i)$.
\end{proof}

Let $S$ be the set of agents $j\neq i$ whose final bundle contains exactly one item of positive value for agent $i$, say $h_j$, and let $T=\{h_j : j\in S\}$. Since $i\notin S$, we have $|S|\le n-1$, and Lemma~\ref{lem:trmms} gives $v_i(M\setminus T)\ge(n-|S|)\cdot\mu_i(n,M)$. On the other hand, $M\setminus T$ consists of $A_i$, the bundles $A_j$ with $j\notin S\cup\{i\}$, and items of value $0$ for agent $i$. By Claim~\ref{cl:mms}, every $A_j$ with $j\notin S\cup\{i\}$ has value at most $2\, v_i(A_i)$; this also holds trivially if $A_j$ contains no item of positive value. With $x=n-|S|\ge 1$, we get
\[
x\cdot\mu_i(n,M)\le v_i(A_i)+2(x-1)\cdot v_i(A_i)=(2x-1)\cdot v_i(A_i),
\]
i.e., $v_i(A_i)\ge\frac{x}{2x-1}\cdot\mu_i(n,M)\ge\frac{n}{2n-1}\cdot\mu_i(n,M)$, as $\frac{x}{2x-1}$ is non-increasing in $x$ and $x\le n$.

Regarding tightness, we provide the following construction. Consider an instance with $n$ agents and $3n-2$ items. We are interested in the view of agent $n$, who values the items as follows:
\begin{equation*}
 v_n(g_j) =
  \begin{cases}
  \frac{n-1}{2n-1} & \text{for}\; 1 \leq j \leq n-1, \\
  \frac{1}{2n-1} & \text{for}\; n \leq j \leq 2n-2, \\
  \frac{n}{2n-1} & \text{for}\; 2n-1 \leq j \leq 3n-2. \\
  \end{cases}
\end{equation*}
Now notice that $\mu_n(n,M)=1$. Indeed, the total value of the items for agent $n$ is $n$, and there is a partition with $n-1$ bundles that contain an item of value $\frac{n-1}{2n-1}$ and an item of value $\frac{n}{2n-1}$, and one bundle that contains one item of value $\frac{n}{2n-1}$ and the $n-1$ items of value $\frac{1}{2n-1}$, so every bundle of this partition has value $1$. All other agents $i\in [n-1]$ have the same valuation function, namely
\begin{equation*}
 v_i(g_j) =
  \begin{cases}
  1 & \text{for}\; 1 \leq j \leq 2n-2, \\
  \epsilon & \text{for}\; 2n-1 \leq j \leq 3n-2, \\
  \end{cases}
\end{equation*}
for some $0<\epsilon<1$. All agents report truthfully, and under Preferred-Item-E-C-E their preference lists order the items in non-increasing order of value, breaking ties by index. In the first step, no agent has an edge, so agent $1$ receives $g_1$. Afterwards, every agent with an empty bundle has an edge towards every agent with a nonempty bundle, so in steps $2,\ldots,n-1$ agents $2,\ldots,n-1$ receive $g_2,\ldots,g_{n-1}$, respectively, and in step $n$ agent $n$ receives her best item $g_{2n-1}$. At this point no agent has an edge: agent $n$ has value $\frac{n}{2n-1}$ for her bundle and $\frac{n-1}{2n-1}$ for each other bundle, and each other agent has value $1$ for her bundle and $\epsilon$ for the bundle of agent $n$. In steps $n+1,\ldots,2n-1$, the items $g_n,\ldots,g_{2n-2}$ are received by agents $1,\ldots,n-1$, respectively: after agent $1$ receives $g_n$, agents $2,\ldots,n-1$ have an edge towards her, and so on, while agent $n$ has value exactly $\frac{n}{2n-1}$ for every other bundle, so she has no edge. Again, at this point no agent has an edge. In steps $2n,\ldots,3n-2$, the items $g_{2n},\ldots,g_{3n-2}$ are received by agents $1,\ldots,n-1$, respectively: after agent $1$ receives $g_{2n}$, all other agents have an edge towards her, and agents $2,\ldots,n-1$ have priority over agent $n$ by index, and so on. No cycle resolution happens during the run, as agent $n$ is never envied and agents $1,\ldots,n-1$ never have an edge towards an agent with a larger index. Thus, in the final allocation, every agent besides agent $n$ has value $\frac{2n}{2n-1}$ according to $v_n$, and agent $n$ has value $\frac{n}{2n-1}=\frac{n}{2n-1}\cdot\mu_n(n,M)$. This completes our proof.
\end{proof}
\end{revisedx}

\begin{corollary}\label{col:trmms}
    \newx{Consider an instance where agent $i$ plays a best response. Under Best-Item-E-C-E and Preferred-Item-E-C-E, if $v_i$ is additive, then $v_i(A_i)\ge\frac{n}{2n-1}\cdot \mu_i(n,M)$. Under any version of E-C-E we consider, if $v_i$ is monotone subadditive, then $v_i(A_i)\ge\frac{1}{n}\cdot \mu_i(n,M)$.}
\end{corollary}

\begin{proof}
\newx{For the first statement, agent $i$ can deviate to the truthful bid of Theorem~\ref{thm:trmms}, which guarantees her value at least $\frac{n}{2n-1}\cdot\mu_i(n,M)$ whatever the other agents report. A best response gives her at least as much.}

\newx{For the second statement, let agent $i$ report truthfully, let $w$ be her final value, and let $T$ be the set of Lemma~\ref{lem:witness}. Since $|T|\le n-1$, some set $P$ of a partition of $M$ into $n$ sets, each of value at least $\mu_i(n,M)$, does not intersect $T$. By monotonicity and Lemma~\ref{lem:witness}, $\mu_i(n,M)\le v_i(P)\le v_i(M\setminus T)\le n\cdot w$. A best response gives agent $i$ value at least $w$.}
\end{proof}

Similarly, the proof of Corollary~\ref{col:trmms} shows that a $\beta$-best response (see Section~\ref{sec:BR4}) guarantees $\beta\cdot\frac{n}{2n-1}\cdot\mu_i(n,M)$ under Best-Item-E-C-E and Preferred-Item-E-C-E for additive $v_i$, and $\frac{\beta}{n}\cdot\mu_i(n,M)$ under any version for monotone subadditive $v_i$.

The proof of Theorem~\ref{thm:trmms} gives a better bound when the final allocation has a favorable structure.

\begin{remark}\label{rem:trmmspar}
In the setting of Theorem~\ref{thm:trmms}, let $S$ be as in its proof, let $L$ be the set of agents $j\ne i$ whose final bundle contains at least two items of positive value for agent $i$, and let $k\ge1$ be such that agent $i$ has been selected at least $k$ times before any bundle contained in some $A_j$, $j\in L$, has value more than $v_i(A_i)$ for her. Then
\[
v_i(A_i)\ge\frac{(n-|S|)\cdot k}{(|L|+1)\cdot k+|L|}\cdot\mu_i(n,M).
\]
(If no such bundle ever has value more than $v_i(A_i)$, any $k\ge1$ can be used.) 
To see the inequality, note that the value of agent $i$ for her own bundle increases by $v_i(g)$ whenever she receives an item $g$, and it does not decrease in cycle resolutions; so the total value of the items she receives when selected is at most $v_i(A_i)$. Every $j\in L$ with $v_i(A_j)\le v_i(A_i)$ already satisfies the bound below. Let $j\in L$ with $v_i(A_j)>v_i(A_i)$, and consider the addition of an item $h$ that makes a bundle contained in $A_j$ have value more than $v_i(A_i)$ for the first time. Before this addition, this bundle had value at most $v_i(A_i)$, and $h$ was available whenever agent $i$ was selected before; so $k\cdot v_i(h)\le v_i(A_i)$. After this addition, agent $i$ has an edge towards the holder of this bundle forever, so, by Observation~\ref{obs:basic}(b), it does not grow. Hence $v_i(A_j)\le v_i(A_i)+\frac1k\cdot v_i(A_i)$ for every $j\in L$. Summing as in the proof of Theorem~\ref{thm:trmms} gives $(n-|S|)\cdot\mu_i(n,M)\le v_i(A_i)+|L|\cdot\left(1+\frac1k\right)\cdot v_i(A_i)$, which is the claimed inequality. As in Corollary~\ref{col:trmms}, the bound applies to a best-responding agent whenever the run under her truthful bid has the stated structure.
\end{remark}

\section{Discussion and Future Directions}\label{sec:dis}

\begin{revisedx}
As our results demonstrate, the sophisticated structure of E-C-E (if one compares it with Round-Robin) gives more freedom to the agents, and prevents it from having pure Nash equilibria even for extremely simple scenarios. On the other hand, E-C-E can be implemented in versions that retain fairness guarantees from the perspective of best-responding agents, as long as there are at most three agents, and these guarantees extend to monotone subadditive valuation functions. For four or more agents, no variant we consider guarantees any approximation of EF1 to every best-responding agent, 
although an $\frac{n}{2n-1}$ fraction of the maximin share is always guaranteed for additive valuation functions when the selected agent gets to choose the assigned item (and a $\frac1n$ fraction, in general). This shows the immense impact that choosing a concrete version of E-C-E has on its properties beyond guaranteeing EF1 fairness.

Several questions remain open. For three agents, we do not know whether our factor of $\frac13$ for Best-Item-E-C-E and Preferred-Item-E-C-E can be improved, nor whether Priority-E-C-E guarantees any approximation of EF1, and for two agents we do not know whether Best-Item-E-C-E guarantees EF1. For four or more agents, our impossibility relies on an agent who is indifferent among all her bids, as well as on reported values equal to $0$, which is unavoidable for Best-Item-E-C-E and Preferred-Item-E-C-E by Proposition~\ref{prop:positive}; it would be interesting to understand whether every agent always has \emph{some} best response that is approximately EF1 from her perspective, or whether a positive best-response value already suffices for such a guarantee. It is also open whether our maximin share guarantees are tight for best-responding, rather than truthful, agents, and whether Preferred-Item-E-C-E always has pure Nash equilibria, at least for two agents. Finally, there are E-C-E-inspired algorithms that combine share-based and envy-based fairness guarantees (e.g., \cite{AkramiR25a}), and it would be interesting to explore to what extent these can be simultaneously maintained in the presence of incentives.
\end{revisedx}

\section*{Acknowledgments}
Georgios Amanatidis was partially supported by the project MIS 5154714 of the National Recovery and Resilience Plan Greece $2.0$ funded by the European Union under the NextGenerationEU Program. Rebecca Reiffenh{\"a}user was partially supported by the NWO grant \emph{Fair Allocation for Strategic Agents (FAStA)}, grant ID https://doi.org/10.61686/AZZOJ38639. \smallskip

\newx{\textbf{Use of AI tools.} 
AI tools were used for checking the correctness of the results, simplifying the proofs of Section \ref{sec:BR}, and helping with the design of the construction of Theorem \ref{thm:four}. The authors used Claude Opus 5.5 and ChatGPT 6.1 Sol. All results and proofs have been separately verified by the authors, who take full responsibility for them. }

\bibliographystyle{plain}
\bibliography{new_refs}

\end{document}